\documentclass[11pt]{article}

\usepackage{algorithm}
\usepackage[noend]{algpseudocode}
\usepackage{amsmath,amssymb,amsthm}
\usepackage{cite}
\usepackage[margin=1in]{geometry}
\usepackage{graphicx}
\usepackage{caption}
\usepackage{subcaption} 
\usepackage{tikz}
\usepackage{xcolor}

\newcommand{\Gtri}{\ensuremath{G_{\Delta}}}
\newcommand{\Geqt}{\ensuremath{G_\text{eqt}}}
\newcommand{\pweight}[1]{\overline{p}(\sigma,#1)}
\newcommand{\cshort}{\nu}
\newcommand{\M}{\mathcal{M}}
\newcommand{\A}{\mathcal{A}}

\newcommand{\sthree}{1.732050808}  
\newcommand{\cbound}{2+\sqrt{2}}   

\newif\ifcomment
\commenttrue    

\newif\iftikz
\tikztrue   

\newtheorem{theorem}{Theorem}[section]
\newtheorem{lemma}[theorem]{Lemma}

\newtheorem{corollary}[theorem]{Corollary}

\newtheorem{property}{Property}

\title{A Stochastic Approach to Shortcut Bridging in Programmable Matter}

\author{Marta {Andr\'es\ Arroyo}$^1$\thanks{{\tt martaandres@correo.ugr.es}.} \and
        Sarah Cannon$^2$\thanks{{\tt sarah.cannon@berkeley.edu}. Supported in part by NSF DGE-1148903 and a grant from the Simons Foundation (\#361047 to Sarah Cannon).} \and
        Joshua J. Daymude$^3$\thanks{{\tt jdaymude@asu.edu}. Supported in part by NSF CCF-1422603, CCF-1637393, and CCF-1733680.} \and
        Dana Randall$^4$\thanks{{\tt randall@cc.gatech.edu}. Supported in part by NSF CCF-1526900, CCF-1637031, and CCF-1733812.} \and
        Andr\'ea W. Richa$^3$\thanks{{\tt aricha@asu.edu}. Supported in part by NSF CCF-1422603, CCF-1637393, and CCF-1733680.}}
\date{\small $^1$ University of Granada, Spain \\ $^2$ Computer Science Division, University of California, Berkeley, Berkeley, CA 94709 \\ $^3$ Computing, Informatics, and Decision Systems Engineering, Arizona State University, Tempe, AZ 85281 \\ $^4$ College of Computing, Georgia Institute of Technology, Atlanta, GA 30332-0765}

\begin{document}

\maketitle
\setcounter{footnote}{0}  

\begin{abstract}
In a \emph{self-organizing particle system}, an abstraction of programmable matter, simple computational elements called \emph{particles} with limited memory and communication self-organize to solve system-wide problems of movement, coordination, and configuration.
In this paper, we consider a stochastic, distributed, local, asynchronous algorithm for ``shortcut bridging'', in which particles self-assemble bridges over gaps that simultaneously balance minimizing the length and cost of the bridge.
Army ants of the genus \emph{Eciton} have been observed exhibiting a similar behavior in their foraging trails, dynamically adjusting their bridges to satisfy an efficiency trade-off using local interactions. 
Using techniques from Markov chain analysis, we rigorously analyze our algorithm, show it achieves a near-optimal balance between the competing factors of path length and bridge cost, and prove that it exhibits a dependence on the angle of the gap being ``shortcut'' similar to that of the ant bridges.
We also present simulation results that qualitatively compare our algorithm with the army ant bridging behavior.
Our work gives a plausible explanation of how convergence to globally optimal configurations can be achieved via local interactions by simple organisms (e.g., ants) with some limited computational power and access to random bits.
The proposed algorithm also demonstrates the robustness of the stochastic approach to algorithms for programmable matter, as it is a surprisingly simple extension of our previous stochastic algorithm for compression. 
\end{abstract}

\section{Introduction} \label{sec:intro}

To develop a system of \emph{programmable matter}, one endeavors to create a material or substance that utilizes user input or stimuli from its environment to change its physical properties in a programmable fashion.
Many such systems have been proposed (e.g., DNA tiles, synthetic cells, and reconfigurable modular robots) and each attempts to perform tasks subject to domain-specific capabilities and constraints.
In our work on \emph{self-organizing particle systems}, we abstract away from specific settings and envision a system of computationally limited devices (which we call \emph{particles}) that can actively move and individually execute distributed, local, asynchronous algorithms to cooperatively achieve macro-scale tasks of movement and coordination.

The phenomenon of local interactions yielding emergent, collective behavior is often found in natural systems; for example, honey bees choose hive locations based on decentralized recruitment~\cite{Camazine1999} and cockroach larvae perform self-organizing aggregation using pheromones with limited range~\cite{Jeanson2005}.
In this paper, we present an algorithm inspired by the work of Reid et al.~\cite{Reid2015}, who found that army ants continuously modify the shape and position of foraging bridges --- constructed and maintained by their own bodies --- across holes and uneven surfaces in the forest floor.
These bridges appear to stabilize in a structural formation that balances the ``benefit of increased foraging trail efficiency'' with the ``cost of removing workers from the foraging pool to form the structure''~\cite{Reid2015}.

We attempt to capture this inherent trade-off in our algorithm for ``shortcut bridging'' in self-organizing particle systems (formally defined in Section~\ref{subsec:problem}).
Our algorithm is an extension of the stochastic, distributed algorithm for \emph{compression} introduced in~\cite{Cannon2016}, demonstrating that many fundamental elements of our stochastic approach can be generalized to applications beyond the specific context of compression, in which a particle system gathers together as tightly as possible.
In particular, this stochastic approach may be of future interest in the molecular programming domain, where simpler variations of bridging have been studied.
Groundbreaking works in this area, such as that of Mohammed et al.~\cite{Mohammed2017}, focus on forming molecular structures that connect some fixed points; our work may offer insights on further optimizing the quality and/or cost of the resulting bridges.

Shortcut bridging is an attractive goal for programmable matter systems, as many application domains envision deploying programmable matter on surfaces with structural irregularities or dynamic topologies.
For example, one commonly imagined application of smart sensor networks is to detect and span small cracks in infrastructure such as roads or bridges; dynamic bridging behavior would enable the system to remain connected and shift position as cracks form.

\subsection{Related Work} \label{subsec:relwork}



When considering recently proposed and realized systems of programmable matter, one can distinguish between \emph{passive} and \emph{active} systems.
In passive systems, computational units cannot control their movements and have (at most) very limited computational abilities, relying instead on their physical structure and interactions with the environment to achieve locomotion (e.g.,~\cite{Woods2015,Angluin2006,Reid2016}).
A large body of research in molecular self-assembly falls under this category, which has mainly focused on shape formation (e.g.,~\cite{Douglas2009,Cheung2011,Wei2012}).
In contrast, our work examines building dynamic bridges whose exact shape is not predetermined.
Mohammed et al.~studied a similar problem of connecting DNA origami landmarks with DNA nanotubes, using a carefully designed process of nanotube nucleation, growth, and diffusion to achieve and maintain the desired connections~\cite{Mohammed2017}.
Significant differences between their approach and ours are: $(i)$ the bridges we consider already connect their endpoints at the start and we focus on the specific goal of optimizing their shape with respect to a parameterized objective function, and $(ii)$ our system is active as opposed to passive.

Active systems are composed of computational units that can control their actions to solve a specific task.
Examples include \emph{swarm robotics}, various other models of modular robotics, and the \emph{amoebot model}, which is our computational framework (detailed in Section~\ref{subsec:model}).

Swarm robotic systems usually involve collections of autonomous robots moving freely in space with limited sensing and communication ranges.
These systems can perform a variety of tasks including gathering~\cite{Cieliebak2012}, shape formation~\cite{Flocchini2008,Rubenstein2014}, and imitating the collective behavior of natural systems~\cite{Chazelle2009}; however, the individual robots typically have more powerful communication and processing capabilities than those we consider.
\emph{Modular self-reconfigurable robotic systems} focus on the motion planning and control of kinematic robots to achieve dynamic morphology~\cite{Yim2007}, and \emph{metamorphic robots} form a subclass of self-reconfiguring robots~\cite{Chirikjian1994} that share some characteristics with our geometric amoebot model.
Walter et al.~have conducted some algorithmic research on these systems (e.g.,~\cite{Walter2004-chains,Walter2004-envelop}), but focus on problems disjoint from those we consider.

In the context of molecular programming, our model most closely relates to the {\it nubot} model by Woods et al.~\cite{Woods2013,Chen2015}, which seeks to provide a framework for rigorous algorithmic research on self-assembly systems composed of active molecular components, emphasizing the interactions between molecular structure and active dynamics.
This model shares many characteristics with our amoebot model (e.g., space is modeled as the triangular lattice, nubot monomers have limited computational abilities, and there is no global orientation) but differs in that nubot monomers can replicate or die and can perform coordinated rigid body movements.
These additional capabilities prohibit direct translation of results under the nubot model to our amoebot model.

\subsection{The Amoebot Model} \label{subsec:model}

Our computational framework is the \emph{amoebot model}~\cite{sops-amoebot}, originally proposed in~\cite{Derakhshandeh2014} as an abstract model for programmable matter designed to enable rigorous algorithmic research on nano-scale systems.
We envision programmable matter as a collection of individual, homogeneous computational elements called \emph{particles}.
The structure of a particle system is represented as a connected subgraph of the infinite, undirected graph $G = (V,E)$, where $V$ is the set of all locations a particle can occupy relative to its structure and $E$ is the set of all atomic movements between locations in $V$.
Each location in $V$ can be occupied by at most one particle at a time.
For shortcut bridging (and many other problems), we further assume the {\it geometric} amoebot model, in which $G = \Gtri$, the {\it triangular lattice}\footnote{Our past works refer to $\Gtri$ as the \emph{equilateral triangular grid graph} $\Geqt$ and the triangular lattice $\Gamma$.} (Figure~\ref{fig:modelgrid}).

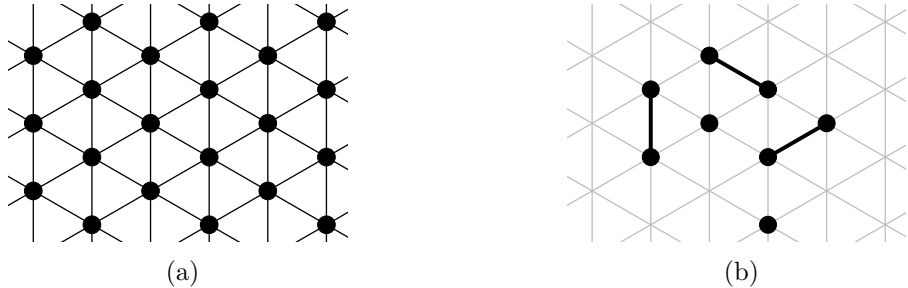
\begin{figure}[t]
\centering
\begin{subfigure}{.45\textwidth}
	\centering
	\iftikz\begin{tikzpicture}[scale=0.9]
    \clip (0.5,-0.25) rectangle (5.5,3.25);
    \foreach \i in {0,...,10} \draw[black,line width=.5pt] (\i*\sthree / 2,-5)--(\i*\sthree / 2,5);
    \foreach \i in {-10,...,10}
    {
    \draw[black,line width=.5pt] (0,\i)--(5*\sthree,\i + 5);
    \draw[black,line width=.5pt] (0,\i)--(5*\sthree,\i - 5);
    }
    \foreach \i in {0,2,...,10}
    \foreach \j in {-5,...,5}
    \draw[fill] (\i*\sthree / 2,\j) circle (0.13);
    \foreach \i in {1,3,...,10}
    \foreach \j in {-4.5,...,4.5}
    \draw[fill] (\i*\sthree / 2,\j) circle (0.13);
\end{tikzpicture}\fi
	\caption{}
	\label{fig:modelgrid}
\end{subfigure}%
\begin{subfigure}{.45\textwidth}
	\centering
	\iftikz\begin{tikzpicture}[scale=0.9]
    \clip (0.5,-0.25) rectangle (5.5,3.25);
    \foreach \i in {0,...,10}
    \draw[lightgray,line width=.5pt] (\i*\sthree / 2,-5)--(\i*\sthree / 2,5);
    \foreach \i in {-10,...,10}
    {
    \draw[lightgray,line width=.5pt] (0,\i)--(5*\sthree,\i + 5);
    \draw[lightgray,line width=.5pt] (0,\i)--(5*\sthree,\i - 5);
    }
    \draw[fill] (1*\sthree,1) circle (0.125);
    \draw[black,line width=1.5pt](1*\sthree,1)--(1*\sthree,2);
    \draw[fill] (1*\sthree,2) circle (0.125);
    \draw[fill] (2*\sthree,1) circle (0.125);
    \draw[black,line width=1.5pt](2*\sthree,1)--(2.5*\sthree,1.5);
    \draw[fill] (2.5*\sthree,1.5) circle (0.125);
    \draw[fill] (1.5*\sthree,2.5) circle (0.125);
    \draw[black,line width=1.5pt](1.5*\sthree,2.5)--(2*\sthree,2);
    \draw[fill] (2*\sthree,2) circle (0.125);
    \draw[fill] (1.5*\sthree,1.5) circle (0.125);
    \draw[fill] (2*\sthree,0) circle (0.125);
\end{tikzpicture}\fi
	\caption{}
	\label{fig:modelparticles}
\end{subfigure}
\caption{(a) A section of the triangular lattice $\Gtri$; (b) expanded and contracted particles.}
\label{fig:model}
\end{figure}

Each particle is either \emph{contracted}, occupying a single location, or \emph{expanded}, occupying a pair of adjacent locations in $\Gtri$ (Figure~\ref{fig:modelparticles}).
Particles move via a series of {\it expansions} and {\it contractions}: a contracted particle may expand into an adjacent unoccupied location, and completes its movement by contracting to once again occupy a single location.
An expanded particle's \emph{head} is the location it last expanded into and the other location it occupies is its \emph{tail}; a contracted particle's head and tail are the same location.

Two particles occupying adjacent locations in $\Gtri$ are said to be {\it neighbors}.
Each particle is {\it anonymous}, lacking a unique identifier, but can locally identify each of its neighboring locations and can determine which of those locations are occupied by particles.
Each particle has a constant-size, local memory that its neighbors can directly read from for communication.
A particle's memory stores whether it is contracted or expanded and identifies if neighboring locations are incident to its head or tail.
Particles do not have access to any global information such as a global compass or an estimate of the size of the system.

We assume the standard asynchronous model from distributed computing (see, e.g.,~\cite{Lynch1996}), where a system progresses through {\it atomic actions}.
A classical result under this model states that for any concurrent asynchronous execution of atomic actions, there is a sequential ordering of actions producing the same end result, provided conflicts that arise in the concurrent execution are resolved.
In our setting, an atomic action is an activation of a single particle.
Once activated, a particle can perform an arbitrary, bounded amount of computation involving its local memory and the memories of its neighbors, and can perform at most one contraction or expansion.
We assume conflicts arising from simultaneous particle expansions into the same unoccupied location are resolved arbitrarily such that at most one particle is expanding into a given location at a time.
Thus, while in reality many particles may be active concurrently, it suffices when analyzing our algorithm to consider a sequence of activations where only one particle is active at a time.

\paragraph{Terminology for Particle Systems}

In addition to the formal model, we introduce some terminology for our application of shortcut bridging.
Just as the uneven surfaces of the forest floor affect the foraging behavior of army ants, the collective behavior of particle systems should change when $\Gtri$ is non-uniform.
Here, we focus on system behaviors when the locations of $\Gtri$ are either \emph{gap} (unsupported) or \emph{land} (supported).
A particle can tell whether its location is a gap location or a land location.
An \emph{object} is a static particle that does not perform computation; these are used to keep the particle system connected to certain fixed sites.

A particle system \emph{configuration} is the finite set of occupied locations of $\Gtri$.
An \emph{edge} of a configuration is an edge of $\Gtri$ where both endpoints are occupied by particle tails\footnote{Lattice edges incident to a node occupied by an expanded particle's head are not counted as configuration edges, since these are exploratory and temporary. This is explained further in Section~\ref{subsec:localalg}.}.
When referring to a \emph{path}, we mean a path of such edges.
Two particles are \emph{connected} if there exists a path between them, and a configuration is \emph{connected} if all pairs of particles are.
A \emph{hole} in a configuration is a maximal finite component of adjacent unoccupied locations.
We specifically consider connected configurations with no holes, and our algorithm --- if starting at such a configuration --- will maintain these properties, a fact we will prove in Section~\ref{subsec:mcprops}.

Let $\sigma$ be a connected configuration with no holes.
The (single, external) \emph{boundary} of $\sigma$ is the walk composed of all edges in $\sigma$ between particles that are not surrounded (i.e., those with less than $6$ neighbors)\footnote{Note that an edge may appear twice in the boundary if it is a cut-edge (e.g., the bottom-left most edges in Figure~\ref{fig:simvshape}).}.
In order to analyze the strength of the solutions our algorithm produces, we define the \emph{weighted perimeter} $\pweight{c}$ to be the summed weight of the edges on the boundary of $\sigma$, where edges between land locations have weight $1$, edges between gap locations have weight $c > 1$, and edges with one endpoint on land and one endpoint in the gap have weight $(1+c)/2$.

\subsection{Problem Description} \label{subsec:problem}

In the \emph{shortcut bridging problem}, we consider an instance $(L,O,\sigma_0,c,\alpha),$ where $L \subseteq V$ is the set of land locations, $O$ is the set of (two) objects to bridge between, $\sigma_0$ is the initial configuration of the particle system, $c > 1$ is a fixed weight for edges between gap locations, and $\alpha > 1$ is a parameter capturing our error tolerance.
An instance is \emph{valid} if $(i)$ the objects of $O$ and particles of $\sigma_0$ all occupy locations in~$L$, $(ii)$ $\sigma_0$ connects the objects, and $(iii)$ $\sigma_0$ is connected.
A (distributed) algorithm \emph{solves} a valid instance $(L, O, \sigma_0, c, \alpha)$ if, beginning from $\sigma_0$, it reaches and remains in a set of configurations $\Sigma^*$ such that any $\sigma \in \Sigma^*$ has weighted perimeter $\pweight{c}$ within an $\alpha$-factor of its minimum possible value, with high probability\footnote{An event occurs \emph{with high probability (w.h.p.)} if the probability of success is at least $1 - 1 / \text{poly}(n)$; here, $n$ is the number of particles.}.

In analogy to the apparatus used in~\cite{Reid2015} (Figure~\ref{fig:apparatus}), we are particularly interested in instances where $L$ forms a V-shape,~$O$ has two objects positioned at either base of $L$, and~$\sigma_0$ lines the interior sides of $L$, as in Figure~\ref{fig:initconfig_V}.
However, our algorithm is not limited to this setting; for example, we show simulation results for an N-shaped land mass (Figure~\ref{fig:initconfig_N}) in Section~\ref{sec:simulations}.

\begin{figure}[t]
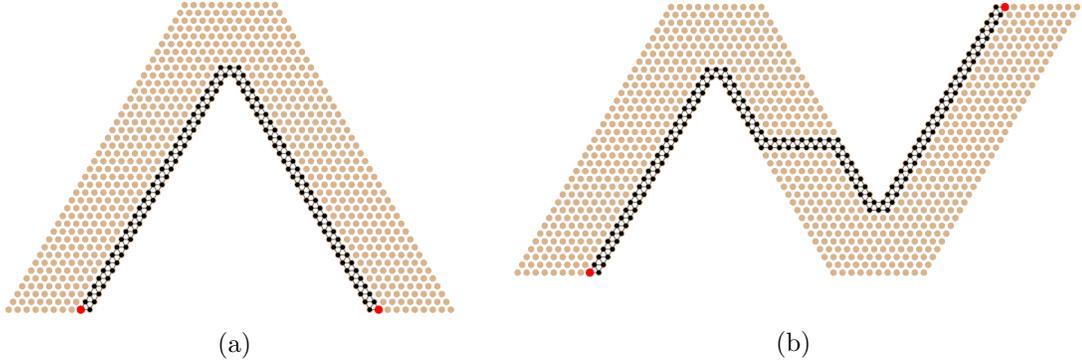

\centering
\begin{subfigure}{.45\textwidth}
	\centering
	\iftikz\input{fig_angle60_lambda4_gamma2_0.tex}\fi
	\caption{}
	\label{fig:initconfig_V}
\end{subfigure}%
\begin{subfigure}{.45\textwidth}
	\centering
	\iftikz\input{fig_N_initial.tex}\fi
	\caption{}
	\label{fig:initconfig_N}
\end{subfigure}
\caption{Example initial configurations $\sigma_0$ of particles (black) connecting two objects $O$ (large, red) on land masses $L$ (brown and black) for two instances of the shortcut bridging problem for which we present simulation results (Section~\ref{sec:simulations}).}
\label{fig:initconfigs}
\end{figure}

The weighted perimeter balances the trade-off observed in~\cite{Reid2015} between the competing objectives of establishing a short path between the fixed endpoints while not having too many particles in the gap.
Although both metrics are amenable to our analysis, we focus on weighted perimeter instead of the number of particles in the gap for two reasons.
First, the structure and thickness of bridges produced using weighted perimeter more closely resemble those of ant bridges, while using particles in the gap results in consistently thin, jagged structures (see Figure~\ref{fig:simvshape} vs.~\ref{fig:simnumingap}).
Second, only particles on the perimeter can move, and thus recognize the potential risk of being in the gap.

\begin{figure}
\centering
\begin{subfigure}{.31\textwidth}
	\centering
	\includegraphics[scale = 0.2, trim = 50 0 100 0, clip]{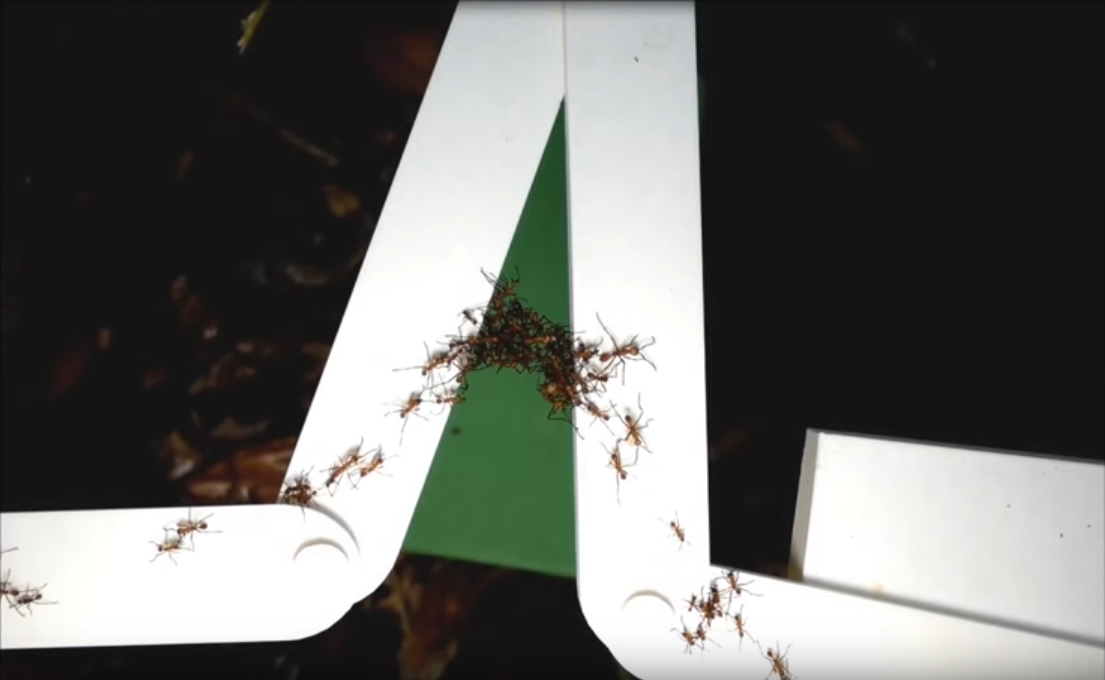}
	\caption{}
	\label{fig:apparatus}
\end{subfigure}%
\begin{subfigure}{.34\textwidth}
	\centering
	\iftikz\input{fig_60_20000000.tex}\fi
	\caption{}
	\label{fig:simvshape}
\end{subfigure}%
\begin{subfigure}{.34\textwidth}
    \centering
    \iftikz\input{fig_numParticlesInGap_10000000.tex}\fi
    \caption{}
    \label{fig:simnumingap}
\end{subfigure}
\caption{(a) In this image from~\cite{Reid2015}, army ants of the genus \emph{Eciton} build a dynamic bridge which balances the benefit of a shortcut path with the cost of committing ants to the structure. (b) Our shortcut bridging algorithm also balances competing objectives and converges to similar configurations. (c) Minimizing the number of particles in the gap instead of the weighted perimeter results in thin bridges with large clusters of particles on land that do not resemble the ant bridges as closely.}
\label{fig:antcomparison}
\end{figure}

\section{Approach, Techniques, and Results} \label{sec:background}

In~\cite{Cannon2016}, we introduced a stochastic, distributed algorithm for compression in the amoebot model; here we extend that work to show our stochastic approach is more widely applicable.

\subsection{The Stochastic Approach to Particle Systems} \label{subsec:approach}

In the stochastic approach to self-organizing particle systems, we use concepts from statistical physics to design our algorithms, a process we outline here.
At a high level, we define an energy function that captures our objectives for the particle system and then design a Markov chain that, in the long run, favors configurations with desirable energy values.
Care is taken to ensure this Markov chain can be executed in a distributed, asynchronous manner by each particle individually.
While understanding our approach and motivation is not necessary for understanding our results, it provides further insights into our methodologies.

In statistical physics, ensembles of particles similar to those we consider represent physical systems and demonstrate that local micro-behavior can induce global macro-scale changes to the system~\cite{Baxter1980,Blanca2018,Restrepo2013}. 
Like a spring relaxing, physical systems favor configurations that minimize energy.
Each configuration $\sigma$ has energy determined by a \emph{Hamiltonian} $H(\sigma)$, and we then assign each a weight $w(\sigma) = e^{-B \cdot H(\sigma)}$, where $B = 1/T$ is inverse temperature. 
Markov chains have been well-studied as a tool for sampling configurations of these systems with probability proportional to $w(\sigma)$, that is, with probability $w(\sigma)/Z$, where $Z = \sum_{\tau} e^{-B \cdot H(\tau)}$ is the normalizing constant known as the \emph{partition function}.
The configurations with the lowest values of $H(\sigma)$ -- those with the least energy -- are most likely to be sampled.

For shortcut bridging, we introduce a Hamiltonian over particle system configurations that assigns the lowest energy values to configurations with desirable bridge structures; we then design our algorithm to favor these configurations with small Hamiltonians.
We assign each configuration~$\sigma$ a Hamiltonian $H(\sigma) = \pweight{c}$, its weighted perimeter.
Setting $\lambda = e^{B}$, we get $w(\sigma,c) = \lambda^{-\pweight{c}}$, where $w(\sigma, c)$ is the likelihood with which we want our algorithm to yield $\sigma$.
As~$\lambda$ gets larger (by increasing $B$, effectively lowering temperature), these weights increasingly favor configurations where $H(\sigma) = \pweight{c}$ is small and the desired bridging behavior is exhibited.
Using a Markov chain, we will ensure that the eventual probability with which we are at state~$\sigma$ is $w(\sigma, c)/Z$, where $Z = \sum_\tau w(\tau,c)$ in the necessary normalizing factor.


\subsection{Markov Chains} \label{subsec:markovchains}

We briefly review relevant terminology on Markov chains.
A {\it Markov chain} $\M$ is a memoryless stochastic process defined on a state space $\Omega$.
We only consider $\Omega$ which are finite and discrete; in particular, the states of $\Omega$ will be connected, hole-free configurations with a common land mass~$L$, objects $O$, and number of particles $n$.
The transition matrix $Q: \Omega \times \Omega \to [0,1]$  of a Markov chain $\M$ is defined so that $Q(\sigma,\tau)$
is the probability of moving from state $\sigma$ to state $\tau$ in one step, for any pair of states $\sigma,\tau \in \Omega$.
For our Markov chain, transitions will correspond to one particle moving one unit in one direction, and the probabilities of these transitions will be chosen carefully.
The $t$-step transition probability $Q^t(\sigma,\tau)$ is the probability of moving from~$\sigma$ to $\tau$ in exactly $t$ steps.

A Markov chain is \emph{irreducible}, or its state space is \emph{connected}, if there is a sequence of valid transitions from any state to any other state, i.e., for all $\sigma,\tau \in \Omega$, there is a $t$ such that $Q^t(\sigma,\tau) > 0$.
A Markov chain is \emph{aperiodic} if for all $\sigma,\tau \in \Omega$, $\gcd\{t : Q^t(\sigma,\tau) > 0\} = 1$.
A Markov chain is \emph{ergodic} if it is both irreducible and aperiodic.
Any finite, ergodic Markov chain converges to a unique {\it stationary distribution} $\pi$ given by, for all $\sigma, \tau \in \Omega$, $\lim_{t \to \infty} Q^t(\sigma,\tau) = \pi(\tau)$.
Any distribution $\pi'$ satisfying $\pi'(\sigma)Q(\sigma,\tau) = \pi'(\tau)Q(\tau,\sigma)$ for all $\sigma,\tau \in \Omega$ (the {\it detailed balance condition}) must be this unique stationary distribution (see, e.g.,~\cite{Feller1968}).

Given a state space $\Omega$, a set of allowable transitions between states, and a desired stationary distribution $\pi$ on $\Omega$ (e.g., $\pi(\sigma) \sim w(\sigma, c)$), the celebrated Metropolis-Hastings algorithm~\cite{Hastings1970} gives a Markov chain on $\Omega$ that uses only allowable transitions and has stationary distribution $\pi$. This is accomplished by carefully setting the probabilities of the state transitions as follows.
Starting at $\sigma \in \Omega$, pick a neighbor $\tau \in \Omega$ (i.e., a state $\tau$ to which $\sigma$ has an allowable transition) uniformly with probability $1/(2\Delta)$, where $\Delta$ is the maximum number of neighbors of any state, and move to~$\tau$ with probability $\min\{1, \pi(\tau)/\pi(\sigma)\}$; with the remaining probability stay at $\sigma$ and repeat.
If the allowable transitions suffice to go between any two states of $\Omega$, then $\pi$ must be the stationary distribution by detailed balance.
While calculating $\pi(\tau)/\pi(\sigma)$ seems to require global knowledge, this ratio can often be calculated using only local information when many terms cancel out.
For shortcut bridging, because our desired stationary distribution will be $\pi(\sigma) = w(\sigma, c)/Z = \lambda^{-\overline{p}(\sigma,{c})}/Z$ where $Z = \sum_{\tau} w(\tau, c)$, the Metropolis-Hastings probabilities can be written as $\min\{1, \lambda^{\overline{p}(\sigma,c) - \overline{p}(\tau,c)}\}$.
Using this probability calculation to decide whether or not to make a transition is a \emph{Metropolis filter}.
Importantly, if $\sigma$ and $\tau$ only differ by one particle $P$, as is the case with all moves of our algorithm, then $\overline{p}(\sigma,c) - \overline{p}(\tau,c)$, the difference in weighted perimeter due to particle $P$'s move, can be calculated using only local information from the neighborhood of $P$ (Lemma~\ref{lem:local}). 

\subsection{Results} \label{subsec:results}

We present a Markov chain $\M$ for \emph{shortcut bridging} in the geometric amoebot model that translates directly to a fully distributed, local, asynchronous algorithm $\A$.
This Markov chain $\M$ uses only local moves and, using a Metropolis filter, eventually reaches a distribution that favors configurations proportional to their weight $w(\sigma,c)$.
Thus, configurations with smaller weighted perimeter $\pweight{c}$ are more likely, as desired.
Rather than terminating $\M$ at some point and using the resulting configuration as a random sample (as is often done with Markov chains) we instead run $\M$ indefinitely, moving among different configurations but remaining at the stationary distribution $\pi$, which we prove meets our desired objectives with high probability.

We prove that $\M$ (and by extension, $\A$) solves the shortcut bridging problem: for any constant $\alpha > 1$, for appropriately chosen values of parameters, the long run probability that $\M$ is in a configuration $\sigma$ with $\pweight{c}$ larger than~$\alpha$ times its minimum possible value is exponentially small.
The key tool used to establish this is a careful \emph{Peierls argument}, used in statistical physics to study non-uniqueness of limiting Gibbs measures and in computer science to establish slow mixing of Markov chains (see, e.g.,~\cite{Levin2009}, Chapter 15).
We then specifically consider V-shaped land masses with an object on each branch of the V, and prove that the resulting bridge structures vary with the interior angle of the V-shaped gap being shortcut --- a phenomenon also observed by Reid et al.~\cite{Reid2015} in the army ant bridges --- and show in simulation that they are qualitatively similar to those of the ants (e.g., Figure~\ref{fig:antcomparison}).


Our shortcut bridging algorithm and others developed with the stochastic approach (e.g.,~\cite{Cannon2016}) have several advantages over other algorithms for programmable matter and self-organizing particle systems.
They are nearly oblivious, only needing to store at most one bit of information between activations.
They are also more robust to failures; even if particles crash and stop moving, our algorithm will converge to the best bridge possible with respect to the crashed particles' fixed locations.
On the other hand, other algorithms for particle systems (e.g.~\cite{Daymude2017,Derakhshandeh2017}) would fail even with a single particle crash.
Finally, our algorithm requires little to no communication between particles.
Since these algorithms are derived from stochastic processes, powerful tools developed to analyze Markov chains can be employed to rigorously understand their behavior.



\section{A Stochastic Algorithm for Shortcut Bridging} \label{sec:algorithm}

Recall that for the shortcut bridging problem, we desire for our algorithm to achieve small weighted perimeter, where boundary edges in the gap cost a factor of $c > 1$ more than those on land.
The algorithm must balance the competing objectives of having a short path between the two objects while not forming too large of a bridge.
We capture these factors by preferring configurations $\sigma$ that have both small \emph{perimeter} $p(\sigma)$, the length of the walk around the boundary of the particle system, and small \emph{gap perimeter} $g(\sigma)$, the number of perimeter edges that are in the gap, where edges with one endpoint in the gap and one endpoint on land count as half an edge in the gap.
While these objectives may appear to be aligned rather than competing, decreasing the length of the overall perimeter increases the gap perimeter and vice versa in the problem instances we consider (e.g., Figure~\ref{fig:initconfigs}).
We note that $\pweight{c} = p(\sigma) + (c-1) g(\sigma)$, and thus minimizing weighted perimeter is equivalent to simultaneously minimizing both perimeter and gap perimeter.

Our Markov chain algorithm incorporates two bias parameters: $\lambda$ and $\gamma$.
The value of $\lambda$ controls the preference for having small perimeter, while $\gamma$ controls the preference for having small gap perimeter.
In this paper, we only consider $\lambda > 1$ and $\gamma > 1$, which correspond to favoring small perimeter and small gap perimeter, respectively.
Using a Metropolis filter, we ensure our algorithm converges to stationary distribution $\pi$ given by $\pi(\sigma) = \lambda^{-p(\sigma)}\gamma^{-g(\sigma)} / Z$ where $Z = \sum_\tau\lambda^{-p(\tau)}\gamma^{-g(\tau)}$ is the normalizing factor necessary to make $\pi$ a probability distribution.
Arithmetic shows:
\[\lambda^{-\pweight{c}} = \lambda^{-p(\sigma) - (c-1) g(\sigma)} = \lambda^{-p(\sigma)} (\lambda^{c-1})^{-g(\sigma)},\]
so setting $\gamma = \lambda^{c-1}$ yields our desired stationary distribution.



We note $\lambda$ is the same parameter that controlled compression in~\cite{Cannon2016}, where particle configurations converged to a distribution proportional to $\lambda^{-p(\sigma)}$.
That work showed that $\lambda > 1$ is not sufficient to ensure compression, so we restrict our attention to $\lambda > \cbound$, the regime where compression provably occurs.

To ensure our algorithm maintains some desired invariants throughout its execution, we introduce two properties every movement must satisfy.
Specifically, these properties maintain system connectivity\footnote{Since particles treat objects as static particles, the particle system may actually disconnect into several components which remain connected through objects.}, prevent holes from forming, and ensure it is possible for our Markov chain to be reversible; more details can be found in~\cite{Cannon2016}.
These last two conditions are necessary for applying established tools from Markov chain analysis.

We use the following notation.
For a location $\ell$, let $N(\ell)$ denote the set of particles and objects\footnote{The notion of location neighborhoods has been extended from~\cite{Cannon2016} to include objects.} adjacent to $\ell$.
For adjacent locations $\ell$ and $\ell'$, we use $N(\ell \cup \ell')$ to denote the set $N(\ell) \cup N(\ell')$, excluding particles or objects occupying $\ell$ or~$\ell'$.
Let $\mathbb{S} = N(\ell) \cap N(\ell')$ be the particles and objects adjacent to both locations; we note $|\mathbb{S}| \in \{0,1,2\}$.
The following properties can be locally checked by an expanded particle occupying $\ell$ and $\ell'$ (e.g., as in Step~\ref{alg:M:conds} of $\M$, Algorithm~\ref{alg:M}), and are symmetric with respect to these locations.

\begin{property} \label{prop:1}
$|\mathbb{S}| \in \{1,2\}$ and every particle or object in ${N(\ell \cup \ell')}$ is connected to a particle or object in $\mathbb{S}$ by a path through $N(\ell \cup \ell')$.
\end{property}
\begin{property} \label{prop:2}
$|\mathbb{S}| = 0$, $\ell$ and $\ell'$ each have at least one neighbor, all particles and objects in $N(\ell) \setminus \{\ell'\}$ are connected by paths within this set, and all particles and objects in $N(\ell') \setminus \{\ell\}$ are connected by paths within this set.
\end{property}

We can now present our Markov chain $\M$ for an instance $(L,O,\sigma_0,c,\alpha)$ of shortcut bridging.
For input parameter $\lambda > \cbound$, set $\gamma = \lambda^{c-1}$.
Beginning at initial configuration $\sigma_0$, which we assume is connected and hole-free\footnote{If $\sigma_0$ has holes, our algorithm will eliminate them and they will not reform~\cite{Cannon2016}; for simplicity, we focus only on the behavior of the system after this occurs.}, repeat the steps of Algorithm~\ref{alg:M}.

\begin{algorithm}
\caption{Markov Chain $\M$ for Shortcut Bridging} \label{alg:M}
\begin{algorithmic}[1]
	\State Choose a particle $P$ uniformly at random (u.a.r.) from all $n$ particles; let $\ell$ be its location. \label{alg:M:select1}
	\State Choose a neighboring location $\ell'$ and $q \in (0,1)$ u.a.r. \label{alg:M:select2}
	\If {$\ell'$ is unoccupied}
	\State $P$ expands to occupy both $\ell$ and $\ell'$. \label{alg:M:expand}
	\State Let $\sigma$ (resp., $\sigma'$) be the configuration with $P$ at $\ell$ (resp., at $\ell'$). \label{alg:M:sigmas}
	    \If {$(i)$ $\ell$ and $\ell'$ satisfy Property~\ref{prop:1} or Property~\ref{prop:2}, $(ii)$ $|N(\ell)| < 5$, and $(iii)$ $q < \lambda^{p(\sigma)-p(\sigma')}\gamma^{g(\sigma)-g(\sigma')}$} $P$ contracts to $\ell'$. \label{alg:M:conds}
	    \Else {} $P$ contracts back to $\ell$. \label{alg:M:back}
	    \EndIf
	\EndIf
\end{algorithmic}
\end{algorithm}

Conditions $(i)$ and $(ii)$ of Step~\ref{alg:M:conds} ensure that the particle system remains connected and no new holes are formed during the execution of $\M$.
In particular, condition $(ii)$ explicitly disallows a particle with five neighbors from moving into the only unoccupied location in its neighborhood, as doing so would create a hole.
Condition $(iii)$ is the Metropolis filter discussed above; the proposed particle move, once confirmed to be valid, only occurs with probability
\[\min\{1, \lambda^{p(\sigma)-p(\sigma')}\gamma^{g(\sigma)-g(\sigma')}\} = \min\{1, \lambda^{\pweight{c} - \overline{p}(\sigma', c)}\},\]
where $\sigma$ is the configuration with $P$ at location $\ell$ and $\sigma'$ is the configuration with $P$ at location $\ell'$.
Although $p(\sigma) - p(\sigma')$ and $g(\sigma) - g(\sigma')$ are values defined at system-level scale, we show these differences can be calculated locally.

\begin{lemma} \label{lem:local}
An expanded particle $P$ occupying adjacent locations $\ell$ and $\ell'$ in $\Gtri$ can calculate the values of $p(\sigma) - p(\sigma')$ and $g(\sigma) - g(\sigma')$ in Step~\ref{alg:M:conds}$(iii)$ of $\M$ using only local information involving $\ell$, $\ell'$, and $N(\ell \cup \ell')$.
\end{lemma}
\begin{proof}
Observe that these values need only be calculated if conditions $(i)$ and $(ii)$ of Step~\ref{alg:M:conds} holds.
By a result of~\cite{Cannon2016},
\[p(\sigma) - p(\sigma') = |N(\ell')| - |N(\ell)|,\]
which can be calculated using only local information.

Recall that gap perimeter is defined as the number of boundary edges in the gap, counting edges between gap and land as half an edge; this is equal to the number of particles that are on the perimeter and in the gap, counted with appropriate multiplicity if a particle appears on the perimeter more than once.
Given a particle $R$ and a configuration $\tau$, let $G(R,\tau)$ be equal to $1$ if $R$ occupies a gap location in $\tau$ and $0$ otherwise.
Let $\delta(R,\tau)$ be the number of times $R$ appears on the perimeter of $\tau$.
Then the desired difference is:
\[g(\sigma) - g(\sigma') = \sum_R\left[G(R,\sigma)\delta(R,\sigma) - G(R,\sigma')\delta(R,\sigma')\right].\]

Define $\Delta(R) = \delta(R,\sigma) - \delta(R,\sigma')$.
For particle $P$, since conditions $(i)$ and $(ii)$ of Step~\ref{alg:M:conds} hold, $\Delta(P) = 0$.
For any particle $R \not\in \{P\} \cup N(\ell \cup \ell')$, $\Delta(R) = 0$ since its neighborhood is not affected by the movement of $P$.
Moreover, for any particle $R \neq P$, $G(R,\sigma) = G(R,\sigma')$ since it does not move. So:
\begin{align*}
g(\sigma) - g(\sigma') &= \delta(P,\sigma)\left[G(P,\sigma) - G(P,\sigma')\right] + \sum_{R \in N(\ell \cup \ell')} G(R,\sigma)\Delta(R).
\end{align*}

The first term is easily calculated locally.
For the summation, it remains to show that $P$ can locally calculate $\Delta(R)$ for any $R \in N(\ell \cup \ell')$.
First suppose that $R$ is occupies a location adjacent to $\ell$ but not $\ell'$. Then:
\[\Delta(R) = \left\{ \begin{array}{cl}
-1 & \text{if $R$ has two neighbors in $N(\ell)$,}\\
1 & \text{if $R$ has no neighbors in $N(\ell)$, and}\\
0 & \text{otherwise.}
\end{array} \right.\]
The opposite is true if $R$ occupies a location adjacent to $\ell'$ but not $\ell$.
Lastly, suppose $R$ occupies a location adjacent to both $\ell$ and $\ell'$. Then:
\[\Delta(R) = \left\{ \begin{array}{cl}
0 & \text{if $R$ has zero or two neighbors in $N(\ell \cup \ell')$,}\\
-1 & \text{if $R$ shares a neighbor with $\ell$ but not $\ell'$, and}\\
1 & \text{if $R$ shares a neighbor with $\ell'$ but not $\ell$.}
\end{array} \right.\]
In all cases, $P$ can calculate $\Delta(R)$, and thus also $g(\sigma) - g(\sigma')$, using only local information.
\end{proof}

The state space $\Omega$ of $\M$ is the set of all configurations reachable from $\sigma_0$ via valid transitions of $\M$.
We conjecture that this includes all connected, hole-free configurations of~$n$ particles connected to both objects, but proving all such configurations are reachable from $\sigma_0$ is not necessary for our results.
(The proof of the corresponding result in \cite{Cannon2016} does not generalize due to the presence of static objects).

\subsection{From $\M$ to a Distributed, Local Algorithm $\A$} \label{subsec:localalg}

In order for individual particles to run $\M$, a Markov chain with centralized control, we must translate $\M$ into a distributed, local, asynchronous algorithm $\A$ that fully respects the constraints of the amoebot model (Section~\ref{subsec:model}).
In particular, the uniformly at random particle selection in Step~\ref{alg:M:select1} of $\M$ must be translated to individual, asynchronous particle activations and a particle's combined expansion and contraction in Steps~\ref{alg:M:expand}--\ref{alg:M:back} of $\M$ must be decoupled into two separate activations because a particle can perform at most one movement per activation.
The remainder of $\M$ can be executed directly in $\A$: Properties~\ref{prop:1} and~\ref{prop:2} are locally verifiable as they only involve a particle's immediate neighborhood, and Lemma~\ref{lem:local} showed that the differences $p(\sigma) - p(\sigma')$ and $g(\sigma) - g(\sigma')$ used in Step~\ref{alg:M:conds} of $\M$ can be calculated locally.
Full details of this construction can be found in~\cite{Cannon2016}.

Under the usual assumptions of the asynchronous model from distributed computing, one cannot assume that the next particle to be activated is equally likely to be any particle, as specified in Step~\ref{alg:M:select1} of $\M$.
To mimic this uniformly random activation sequence in a local way, we assume each particle has its own Poisson clock with mean $1$ and activates after a delay $t$ drawn with probability $e^{-t}$.
After completing its activation, a new delay is drawn to its next activation, and so on.
The exponential distribution guarantees that, regardless of which particle has just activated, all particles are equally likely to be the next to activate (see, e.g.,~\cite{Feller1968}).
We could even better approximate asynchronous activation sequences by allowing each particle to have its own constant mean for its Poisson clock, allowing for some particles to activate more often than others in expectation.
In this setting, the probability that a particle $P$ is the next of the $n$ particles to activate is not $1/n$, but rather some probability $a_P$ that depends on all particles' Poisson means\footnote{Probability $a_P$ only plays a role in the analysis of $\A$ and $\M$, not in their execution. Particle $P$ does not need to know or calculate $a_P$.}.
This does not change the stationary distribution of $\M$; Lemma~\ref{lem:stationary1} still holds with a nearly identical proof that replaces $1/n$ with $a_P$, and Lemma~\ref{lem:stationary2} and Theorem~\ref{thm:alphabridge} still follow.
Because the same results hold regardless of the rates of particles' Poisson clocks, we assume clocks with mean $1$ for simplicity.

Unlike in $\M$, the amoebot model assumes a particle~$P$ can perform at most one movement per activation (Section~\ref{subsec:model}), so we must decouple $P$'s movement in one iteration of $\M$, which includes both an expansion and a contraction, into two activations.
However, due to asynchrony, other particles may expand into $P$'s neighborhood after it has expanded but before it contracts.
We utilize flag-locking mechanisms to ensure $P$ retains consistent snapshots of its neighborhood regardless of the movements of other particles between its activations.
When $P$ expands from location $\ell$ to also occupy neighboring location $\ell'$ (Step~\ref{alg:M:expand} of $\M$), it sets a Boolean flag~$f$ to \textsc{True} if it is the only expanded particle in its neighborhood, and to \textsc{False} otherwise.
When $P$ is later activated again, it checks its flag: if $f$ is \textsc{False}, it simply contracts back to its original position $\ell$ since some other particle in its neighborhood activated and expanded earlier.
Otherwise,~$P$ checks the conditions of Step~\ref{alg:M:conds} of $\M$ (ignoring any expanded heads, see the next paragraph) and decides whether to contract to $\ell$ or $\ell'$ accordingly.
Particle $P$ then resets $f$ to \textsc{False} and completes its second activation.
This ensures that at most one particle per neighborhood moves at a time, mimicking the sequential nature of $\M$.

Some explanation is warranted on how particle $P$ identifies expanded heads in its neighborhood and why it ignores them when checking the conditions of Step~\ref{alg:M:conds} of $\M$.
Recall from Section~\ref{subsec:model} that a particle stores whether it is expanded or contracted and which neighboring locations are adjacent to its head in memory.
Particle $P$ can read this information from its neighbors to identify expanded heads in its neighborhood.
Moreover, for particle $P$ to reach Step~\ref{alg:M:conds} of $\M$, its flag $f$ must be set to \textsc{True}.
Any other particle $Q$ that expands into the neighborhood of $P$ must then set its flag to \textsc{False}, since it observes $P$ is already expanded.
Thus, $P$ should ignore the heads of these expanded neighbors, since it is only a matter of time before they are activated again and simply contract their expanded heads.

We have shown our Markov chain $\M$ can be translated into a distributed, local, asynchronous algorithm $\A$, but such an implementation is not always possible in general.
Any Markov chain for particle systems that relies on non-local particle moves or has transition probabilities that rely on non-local information cannot be executed by a local, distributed algorithm.
Moreover, many algorithms under the amoebot model are not stochastic and thus cannot be meaningfully described as Markov chains; see, e.g.~\cite{Daymude2017,Derakhshandeh2017}.

\subsection{Properties of Markov Chain $\M$} \label{subsec:mcprops}

We now show some useful properties of Markov chain~$\M$.
Our first two claims follow from work in~\cite{Cannon2016} and basic properties of Markov chains and our particle systems.

\begin{lemma} \label{lem:connectholes}
If $\sigma_0$ is connected and has no holes, then at every iteration of $\M$, the current configuration is connected and has no holes.
\end{lemma}
\begin{proof}
Cannon et al.~\cite{Cannon2016} proved that no moves allowed in their compression algorithm could introduce holes or disconnect the particle system.
Since the moves allowed by $\M$ are a subset of those in the compression algorithm (since the local properties checked at each iteration are the same), $\M$ cannot introduce holes or disconnect the system.
\end{proof}

\begin{lemma} \label{lem:ergodic}
If $\sigma_0$ has no holes, then $\M$ is ergodic.
\end{lemma}
\begin{proof}
Markov chain $\M$ is irreducible because we defined~$\Omega$ to be precisely those configurations reachable by valid transitions of~$\M$ starting from $\sigma_0$.
$\M$ is aperiodic because at each iteration there is a probability of at least $1/6$ that no move occurs, as each particle has at least one neighbor.
Thus, the chain~$\M$ is ergodic.
\end{proof}

As $\M$ is finite and ergodic, it converges to a unique stationary distribution, and we can find that distribution using detailed balance.

\begin{lemma} \label{lem:stationary1}
The stationary distribution of $\M$ is
\[\pi(\sigma) = \lambda^{-p(\sigma)} \gamma^{-g(\sigma)}/Z,\]
where $Z = \sum_{\sigma' \in \Omega} \lambda^{-p(\sigma')} \gamma^{-g(\sigma')}$.
\end{lemma}
\begin{proof}
Properties~\ref{prop:1} and~\ref{prop:2} ensure that particle $P$ moving from location $\ell$ to location $\ell'$ is valid if and only if $P$ moving from $\ell'$ to $\ell$ is.
This implies for any configurations $\sigma$ and~$\tau$, $Q(\sigma, \tau) > 0$ if and only if $Q(\tau, \sigma) > 0$.
Using this, we easily verify the lemma via detailed balance.

Let $\sigma,\tau \in \Omega$ be distinct configurations that differ by one valid move of a particle $P$ from location $\ell$ to neighboring location $\ell'$, and let $n$ be the number of particles.
Then,
\begin{align*}
Q(\sigma,\tau) &= \frac{1}{n} \cdot \frac{1}{6} \cdot \min\{\lambda^{p(\sigma) - p(\tau)}\gamma^{g(\sigma) - g(\tau)}, 1\}, \text{ and}\\
Q(\tau,\sigma) &= \frac{1}{n}\cdot \frac{1}{6}\cdot \min\{\lambda^{p(\tau) - p(\sigma)}\gamma^{g(\tau) - g(\sigma)}, 1\}.
\end{align*}

Without loss of generality, assume that $\lambda$ and $\gamma$ satisfy $\lambda^{p(\sigma) - p(\tau)}\gamma^{g(\sigma) - g(\tau)} \leq 1$. Then,
\[\pi(\sigma)Q(\sigma,\tau) = \frac{\lambda^{-p(\sigma)}\gamma^{-g(\sigma)}}{Z} \cdot \frac{\lambda^{p(\sigma) - p(\tau)}\gamma^{g(\sigma) - g(\tau)}}{6n}
= \frac{\lambda^{-p(\tau)}\gamma^{-g(\tau)}}{Z} \cdot \frac{1}{6n}
= \pi(\tau)Q(\tau,\sigma).\]

The definition of $Z$ implies $\pi$ satisfies $\sum_{\sigma' \in \Omega}\pi(\sigma') = 1$, so $\pi$ is a valid probability distribution and we conclude $\pi$ is the unique stationary distribution of $\M$.
\end{proof}

The stationary distribution can be alternately expressed using weighted perimeter.

\begin{lemma} \label{lem:stationary2}
For $c = 1 + \log_{\lambda}\gamma$, the stationary distribution of $\M$ is given by \[\pi(\sigma) = \lambda^{-\pweight{c}} / Z,\]
where $Z = \sum_{\sigma' \in \Omega} \lambda^{-\overline{p}(\sigma',c)}$.
\end{lemma}
\begin{proof}
This follows from the definition of $\pweight{c}$.
\end{proof}

\begin{theorem} \label{thm:alphabridge}
Consider an execution of Markov chain $\M$ on state space $\Omega$, with $\lambda > \cbound =: \cshort$ and $\gamma > 1$, where starting configuration $\sigma_0$ has $n$ particles.
For any constant $\alpha$ satisfying
\[\alpha > \frac{\log\lambda}{\log\lambda - \log\cshort} > 1,\]
the probability that a particle configuration $\sigma$ drawn at random from $\M$'s stationary distribution $\pi$ satisfies \[\overline{p}(\sigma, 1+\log_{\lambda}\gamma) > \alpha \cdot \overline{p}_{min}\] is exponentially small in $n$ for sufficiently large $n$, where $\overline{p}_{min}$ is the minimum weighted perimeter of a configuration in $\Omega$.
\end{theorem}
\begin{proof}
This proof mimics that of $\alpha$-compression in~\cite{Cannon2016}, but additional insights and care are necessary to accommodate the difficulties introduced by considering weighted perimeter instead of perimeter.
Throughout we consider weighted perimeter $\overline{p}(\sigma) = \pweight{1 + \log_\lambda \gamma}$.

Define the weight of a configuration $\sigma \in \Omega$ to be:
\[w(\sigma) := \pi(\sigma) \cdot Z = \lambda^{-p(\sigma)}\gamma^{-g(\sigma)} = \lambda^{-\overline{p}(\sigma)},\]
where $Z = \sum_{\sigma' \in \Omega} \lambda^{-p(\sigma')} \gamma^{-g(\sigma')}$.
For a set of configurations $S \subseteq \Omega$, we define its weight $w(S) = \sum_{\sigma \in S} w(\sigma)$; analogously, let $\pi(S) = \sum_{\sigma \in S} \pi(\sigma) = w(S) / Z$.
Let $\sigma_{min} \in \Omega$ be a configuration with minimal weighted  perimeter $\overline{p}_{min}$, and let $S_\alpha$ be the set of configurations with weighted perimeter at least $\alpha \cdot \overline{p}_{min}$. We show that for sufficiently large $n$,
\[\pi(S_\alpha) = \frac{w(S_\alpha)}{Z} < \frac{w(S_\alpha)}{w(\sigma_{min})} \leq \zeta^{\sqrt{n}},\]
where $\zeta < 1$.
The first equality and inequality follow directly from the definitions of $Z$, $w$, and $\sigma_{min}$.
We focus on the last inequality.

Stratify $S_\alpha$ into sets of configurations that have the same weighted perimeter; there are at most $O\left(n^2\right)$ such sets, as the total perimeter and gap perimeter can each take on at most $O(n)$ values.
Label these sets as $A_1, A_2,$ ..., $A_m$ in order of increasing weighted perimeter, where $m$ is the total number of distinct weighted perimeters of configurations in $S_\alpha$.
Let~$\overline{p}_i$ be the weighted perimeter of all configurations in set~$A_i$; since $A_i \subseteq S_\alpha$, then $\overline{p}_i \geq \alpha \cdot \overline{p}_{min}$.

Note $w(\sigma) = \lambda^{-\overline{p}_i}$ for every $\sigma \in A_i$, so to bound $w(A_i)$ it suffices to bound $|A_i|$.
A configuration with weighted perimeter $\overline{p}_i$ has perimeter $p \leq\overline{p}_i$, and a result from~\cite{Cannon2016} that exploits a connection to self-avoiding walks in the hexagon lattice~\cite{Duminil-Copin2012} implies the number of connected, hole-free particle configurations with perimeter $p$ is at most $f(p)\cshort^p$, for some subexponential function $f$.
Letting $p_{min}$ denote the minimum possible (unweighted) perimeter of a configuration of $n$ particles, we conclude that:
\[w(A_i) = \lambda^{-\overline{p}_i} |A_i| \leq \lambda^{-\overline{p}_i} \cdot \sum_{p = p_{min}}^{\overline{p}_i} f(p) \cshort^p \leq \lambda^{-\overline{p}_i} f_1(\overline{p}_i) \cshort^{\overline{p}_i},\]
where $f_1(\overline{p}_i) = \sum_{p = p_{min}}^{\overline{p}_i} f(p)$ is necessarily also a subexponential function because it is a sum of at most a linear number of subexponential terms.
So,
\[w(S_\alpha) = \sum_{i=1}^m w(A_i) \leq \sum_{i=1}^m f_1(\overline{p}_i) \left(\frac{\cshort}{\lambda}\right)^{\overline{p}_i} \leq f_2(n) \left(\frac{\cshort}{\lambda}\right)^{\alpha \cdot \overline{p}_{min}},\]
where $f_2(n) = \sum_{i=1}^m f_1(\overline{p}_i)$ is a subexponential function because $\overline{p}_i = O(n)$, $m = O\left(n^2\right)$, and $f_1$ is subexponential.
The last inequality above holds as $\lambda > \cshort$ and $\overline{p}_i \geq \alpha \cdot \overline{p}_{min}$.  
Then, since $w(\sigma_{min}) = \lambda^{-\overline{p}_{min}}$,
\[\pi(S_\alpha) < \frac{w(S_\alpha)}{w(\sigma_{min})}
\leq f_2(n) \left(\frac{\cshort}{\lambda}\right)^{\alpha \cdot \overline{p}_{min}} \lambda^{\overline{p}_{min}}
= f_2(n) \left[\lambda \left(\frac{\cshort}{\lambda}\right)^\alpha \right]^{\overline{p}_{min}}.\]

The constant $\lambda(\cshort / \lambda)^\alpha$ is less than one whenever $\alpha > \frac{\log\lambda}{\log\lambda - \log\cshort}$.
Since the perimeter of any configuration of $n$ particles is at least $\sqrt{n}$, $\overline{p}_{min} \geq \sqrt{n}$.
Because $f_2(n)$ is subexponentially large but $(\lambda(\cshort / \lambda)^\alpha)^{\sqrt{n}}$ is exponentially small, asymptotically the latter term dominates and we conclude there exists $\zeta < 1$ such that for all sufficiently large $n$,
\[\pi(S_\alpha) < f_2(n) (\lambda(\cshort / \lambda)^\alpha)^{\sqrt{n}} < \zeta^{\sqrt{n}},\]
which proves the theorem.
\end{proof}

Though Theorem~\ref{thm:alphabridge} is proved only in the case where the number of particles is sufficiently large, we expect and observe it to hold for much smaller $n$.
However, we are unable to compute an explicit bound on how large $n$ must be for these results to hold because the exact form of the subexponential function $f(p)$ in the above proof is unknown (see Section 4 of~\cite{Duminil-Copin2012} and references therein).

The following corollary shows that our algorithm solves any instance $(L,O,\sigma_0, c,\alpha)$ of the shortcut bridging problem when parameters $\lambda$ and $\gamma$ are chosen accordingly.

\begin{corollary} \label{cor:correctness}
The distributed, local algorithm $\A$ associated with Markov chain $\M$ solves any valid instance of the shortcut bridging problem where the number of particles is sufficiently large.
\end{corollary}
\begin{proof}
Given any valid instance $(L,O,\sigma_0,c,\alpha)$ of the shortcut bridging problem, it suffices to run $\A$ starting from configuration $\sigma_0$ with parameters $\lambda > (\cbound)^{\frac{\alpha}{\alpha - 1}}$ and ${\gamma = \lambda^{c-1}}$.
Then $\alpha > \frac{\log(\lambda)}{\log(\lambda) - \log(\cbound)} > 1$, so by Theorem~\ref{thm:alphabridge} the system reaches and remains with all but exponentially small probability in a set of configurations with weighted perimeter $\pweight{c} \leq \alpha \cdot \overline{p}_{min}$, where $\overline{p}_{min}$ is the minimum weighted perimeter of a configuration in $\Omega$.
Solving the shortcut bridging problem only requires the weaker condition that this occurs with all but a polynomially small probability, which our algorithm certainly achieves.
\end{proof}


\section{Dependence on Gap Angle} \label{sec:proofs}

To understand the relationship between bridging and shape, we consider V-shaped land masses of various angles (e.g., Figure~\ref{fig:initconfig_V}).
We prove our shortcut bridging algorithm has a dependence on the internal angle $\theta$ of the gap similar to that of the army ant bridges studied by Reid et al.~\cite{Reid2015}.
We show that when $\theta$ is sufficiently small, with all but exponentially small probability the bridge constructed by the particles stays close to the bottom of the gap (away from the apex of angle $\theta$).
On the other hand, we show that for some large values of $\theta$, when $\lambda$ and $\gamma$ satisfy certain conditions, with all but exponentially small probability the bridge stays close to the top of the gap.
We prove these results with a Peierls argument and careful analysis of the geometry of the gap.
Simulations of our shortcut bridging algorithm for varying angles can be found in Section~\ref{sec:simulations}.

We first give a formal construction for the V-shaped land mass $L$ given any $\theta \in (0,\pi)$ and constant width $w \geq 2$.
Let $e \in E$ be any edge of the triangular lattice and label its endpoints as $v_1$ and $v_2$.
Extend line segment $\ell_1$ from $v_1$ such that it forms an angle of $\pi/2 + \theta/2$ with $e$.
Similarly extend line segment $\ell_2$ from $v_2$, of the same length and on the same side of $e$ as $\ell_1$, also forming an angle of $\pi/2 + \theta/2$ with~$e$.
Segments $\ell_1$ and $\ell_2$ then differ in their orientation by angle $\theta$.
Without loss of generality, we assume $\ell_1$ is clockwise from $\ell_2$ around $e$.
Let $b$ be the line through $\ell_1$ and $\ell_2$'s other endpoints (not $v_1$ and $v_2$).
The land mass consists of $v_1$, $v_2$, and all vertices of $\Gtri$ that are outside of $\ell_1$ and $\ell_2$ and from which there exists a lattice path of length at most $w$ to a vertex strictly between $\ell_1$ and $\ell_2$.
Vertices of $\Gtri$ on the opposite side of $b$ from $e$ are not included in the land mass.
For example, Figure~\ref{fig:V_shape30} depicts a land mass with $\theta \sim \pi/6$ and Figure~\ref{fig:V_shape90} shows another with $\theta \sim \pi/2$; both have width $w = 5$.
This careful definition involving edge $e$ is necessary to ensure there are no adjacent land locations on opposite sides of the gap, as could happen for small~$\theta$ if the land mass is not constructed carefully.

\begin{figure}
\centering
\begin{subfigure}{.45\textwidth}
	\centering
	\includegraphics[scale = 0.45]{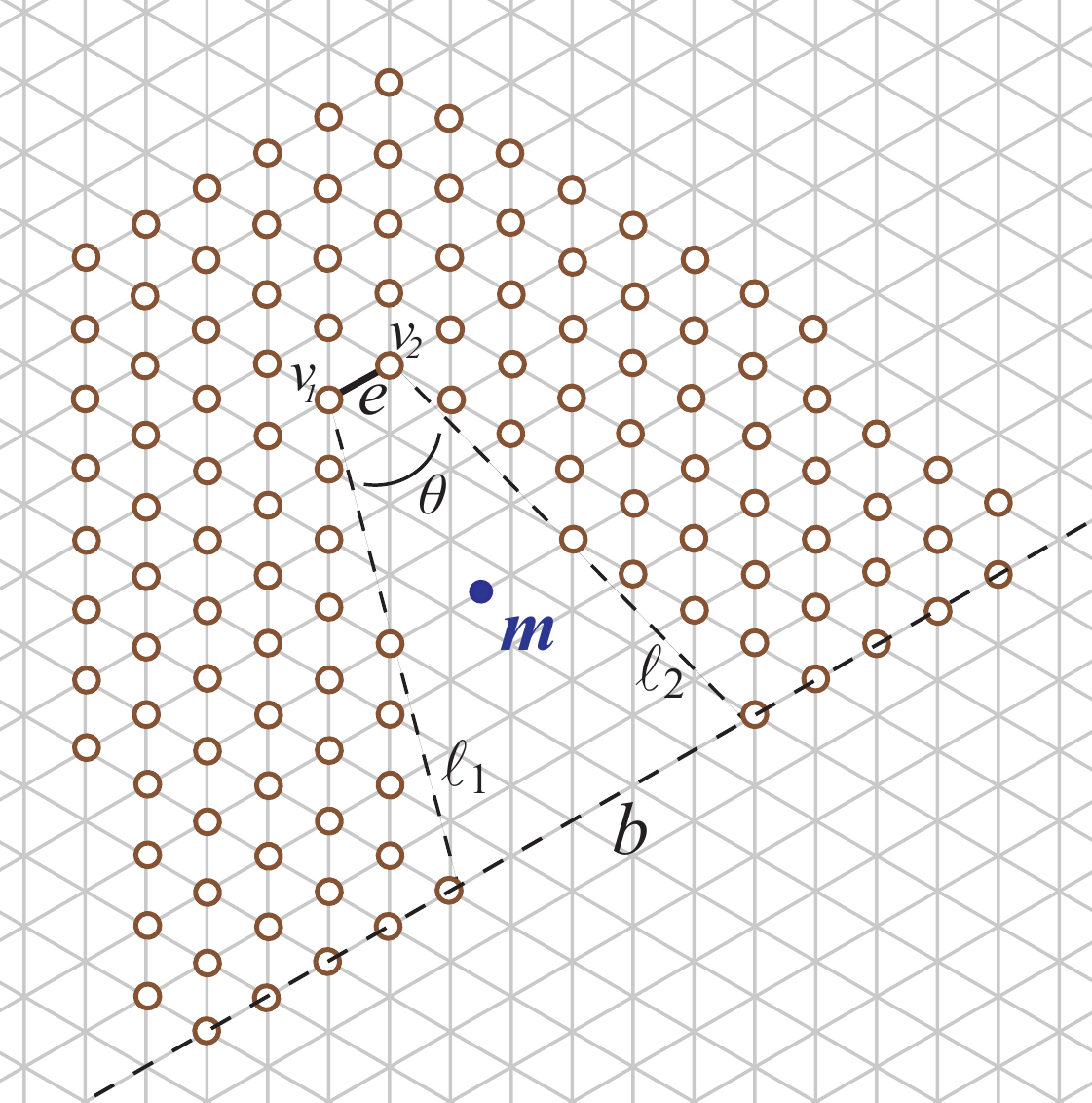}
	\caption{}
	\label{fig:V_shape30}
\end{subfigure}%
\begin{subfigure}{.45\textwidth}
	\centering
	\includegraphics[scale = 0.45]{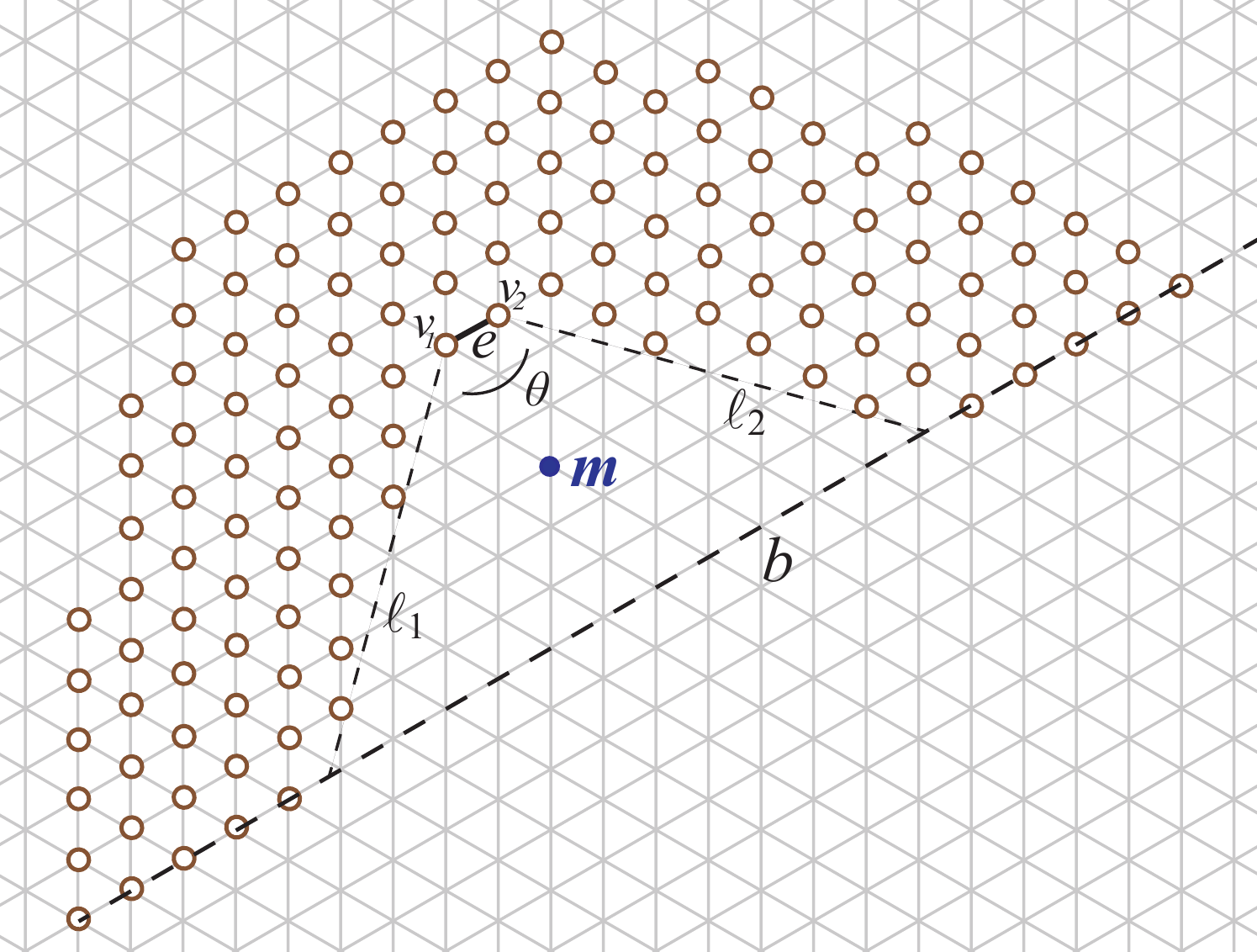}
	\caption{}
	\label{fig:V_shape90}
\end{subfigure}\\ \vspace{3mm}
\begin{subfigure}{.45\textwidth}
	\centering
	\includegraphics[scale = 0.45]{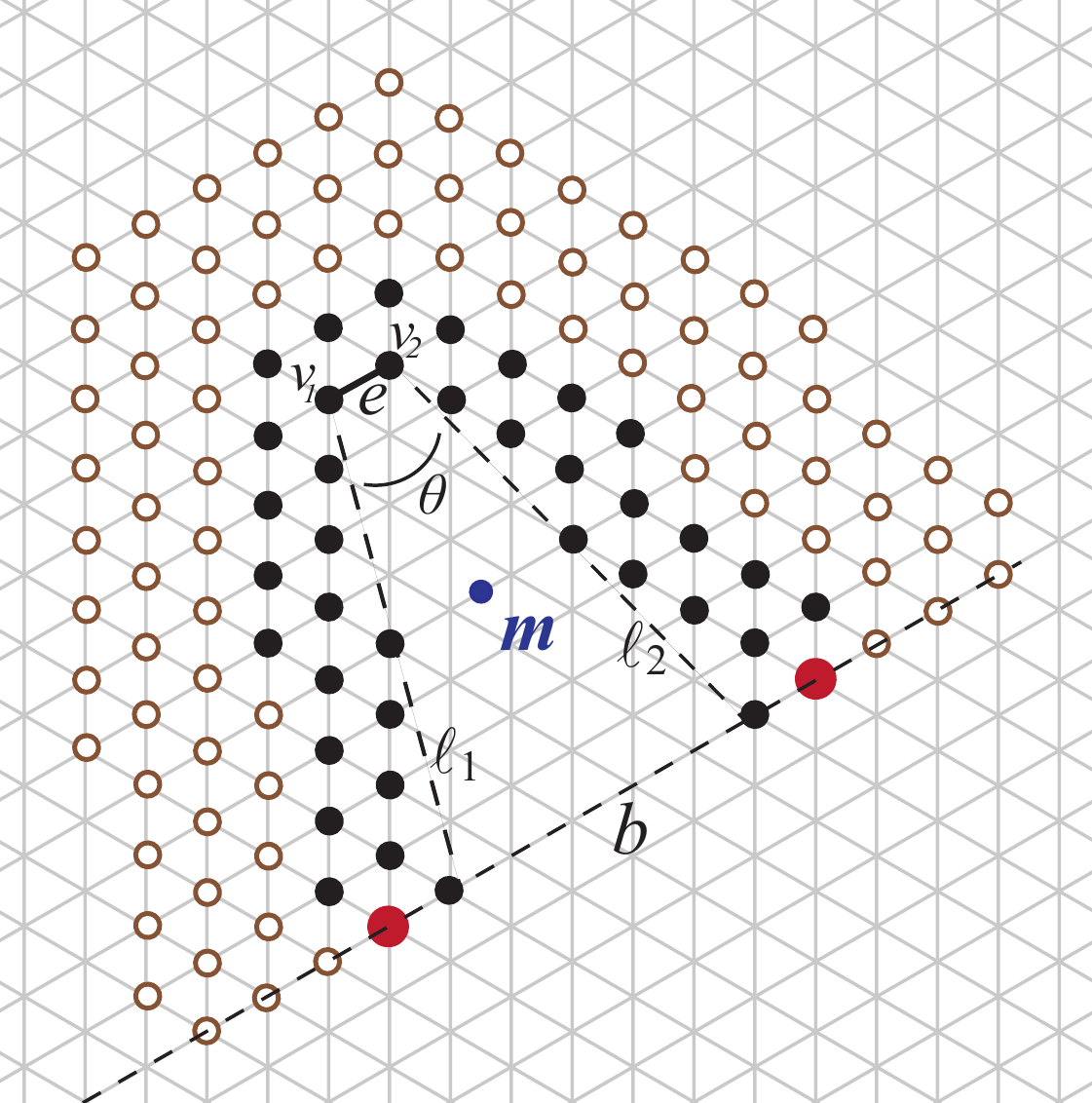}
	\caption{}
	\label{fig:V_shape_initial30}
\end{subfigure}%
\begin{subfigure}{.45\textwidth}
	\centering
	\includegraphics[scale = 0.45]{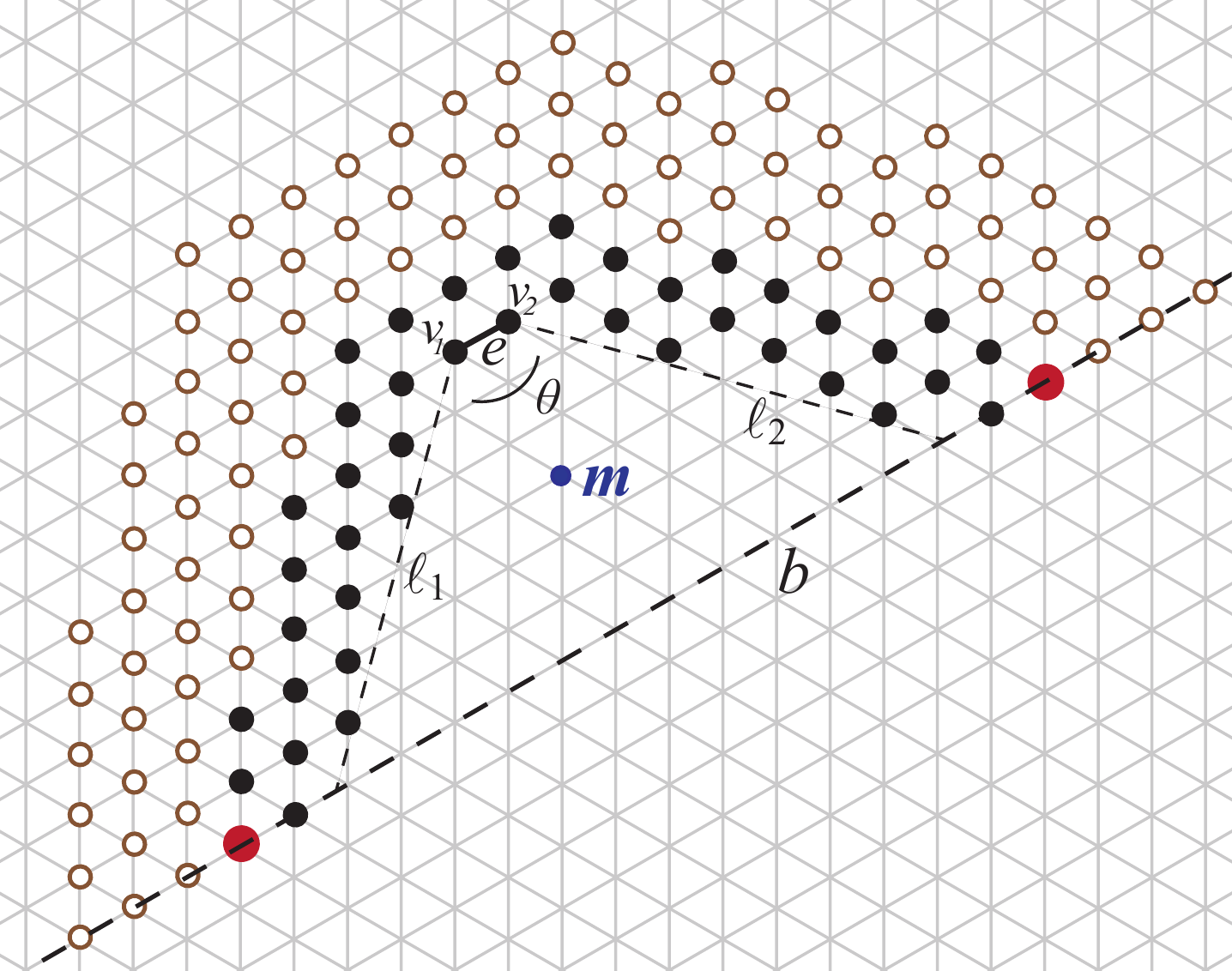}
	\caption{}
	\label{fig:V_shape_initial90}
\end{subfigure}
\caption{The land mass $L$ of constant width $5$ for (a) a small value of $\theta \sim \pi/6$ and height $8$ and (b) a large value of $\theta \sim \pi/2$ and height $9$. The initial configuration $\sigma_0$, with particles shown in black and objects enlarged and red, for (a) a small value of $\theta \sim \pi/6$ and (b) a large value of $\theta \sim \pi/2$. Point $m$ is the midpoint of the segment between the midpoints of $\ell_1$ and $\ell_2$, and $b$ is shown as a dashed line.}
\label{fig:V_shapes}
\end{figure}

From now on we will, in a slight abuse of notation, refer to the gap locations between $\ell_1$ and $\ell_2$ as \emph{the gap}.
By the \emph{bottom of the gap}, we mean the line $b$ through $\ell_1$ and $\ell_2$'s other endpoints (not $v_1$ and $v_2$).
We may assume $b$ is a line of the triangular lattice by truncating $\ell_1$ and $\ell_2$ so that both end on a lattice line; this does not change the land mass $L$.
We also assume $b \cap \ell_1$ and $b \cap \ell_2$ are not vertices of the triangular lattice $\Gtri$; if they are, we can perturb $\ell_1$ and $\ell_2$ slightly, without changing the land mass.
Note $b$ is always parallel to~$e$.

The \emph{height} of land mass $L$ is the length of a shortest path in $\Gtri$ from $v_1$ or $v_2$ to $b$ that only visits land locations; the land mass in Figure~\ref{fig:V_shape30} has height 8, while the land mass in Figure~\ref{fig:V_shape90} has height 9.
Let $m$ be the midpoint of the segment connecting the midpoints of $\ell_1$ and $\ell_2$; $m$ is in the center of the gap, halfway between $e$ and $b$.

The initial configuration $\sigma_0$ we consider is a path of width 2 lining the interior sides of the land mass $L$; see Figures~\ref{fig:V_shape_initial30}--\ref{fig:V_shape_initial90}.
We position the two fixed objects of $O$ in line $b$ at the second vertices outside $\ell_1$ and $\ell_2$, anchoring the particles on either side of the gap.
Note the height of $L$ is exactly the number of particles in $\sigma_0$ next to $\ell_1$ (or $\ell_2$), excluding $v_1$ and $v_2$.

\begin{lemma} \label{lem:V_shape_initial_n}
Let $L$ be a V-shaped land mass of height $k$ and angle $\theta$. The initial configuration $\sigma_0$ has $4k + 5$ particles and two objects.
\end{lemma}
\begin{proof}
First, suppose $\theta \leq \pi/3$, as in Figure~\ref{fig:V_shape_initial30}.
Each lattice line parallel to $e$ and intersecting $\ell_1$ and $\ell_2$, up to but not including $b$, contains exactly four particles.
There are $k$ such lattice lines.
Line $b$ contains two particles.
In the lattice line above and parallel to $e$, there are three particles.
In total, this gives $4k + 2 + 3 = 4k + 5$ particles and two objects.

Now, suppose $\theta > \pi/3$, as in Figure~\ref{fig:V_shape_initial90}; a different counting approach is necessary.
Consider the lattice line through~$v_1$ and the gap location adjacent to $v_1$ and $v_2$; this line and all lines parallel to it intersecting $\ell_1$ contain exactly two particles, and there are $k$ such lines.
The same is true for $v_2$ and $\ell_2$.
Uncounted by this approach are five additional particles: the two particles adjacent to each of the two objects, and the particle adjacent to $v_1$ and~$v_2$.
In total, this gives $2k + 2k + 4 + 1 = 4k + 5$ particles and two objects.
\end{proof}

For a given $\sigma$, let $x$ be the particle or object contained in line $b$ farthest outside of $\ell_1$, and let $y$ be the particle or object in line $b$ farthest outside of $\ell_2$.
We will refer to the perimeter of $\sigma$ traversed counterclockwise from $x$ to $y$ as the \emph{inner perimeter} of $\sigma$. 
We say the inner perimeter is \emph{above a point $p$} if $p$ is to the right of the inner perimeter traversed from $x$ to $y$; it is \emph{below a point $p$} if $p$ is to its left.

We can partition $\Omega$ into two sets $S_1$ and $S_2$, where $S_1$ contains all configurations whose inner perimeter is strictly above midpoint $m$ of the gap and $S_2$ contains all configurations whose inner perimeter goes through or below $m$.
We first prove that for $\lambda > \cbound$ (i.e., in the range of compression) and $\gamma > 1$, there is an angle $\theta_1$ such that for all $\theta < \theta_1$, $\pi(S_1)$ is exponentially small.
We then prove that for $\lambda > \cbound$ and $\gamma > \lambda^4(\cbound)^4$, there is a $\theta_2$ such that for all $\theta \in (\pi/3, \theta_2)$, $\pi(S_2)$ is exponentially small.
We expect much better bounds $\theta_1$ and $\theta_2$ can be obtained with more effort, and that these results generalize to all $\lambda > \cbound$ and $\gamma >1$, but here we simply demonstrate it is possible to give rigorous results about the dependence of the bridge structure on $\theta$.

\subsection{Proofs for Small $\theta$} \label{subsec:smalltheta}

We begin with some structural lemmas.

\begin{lemma} \label{lem:inperimabovelen}
Let $L$ be a V-shaped land mass of height $k$ and angle $\theta \leq \pi/3$. Then any path in $\Gtri$ that starts and ends at the bottom of the gap and goes strictly above the midpoint $m$ of the gap has length at least $k+1$.
\end{lemma}
\begin{proof}
For $\theta \leq \pi/3$, there are $k-1$ lattice lines parallel to~$b$ strictly between $b$ and $e$.
Of these lines exactly $\lceil (k-1)/2 \rceil$ are below or contain $m$.
Any path from $b$ to a location above~$m$ and back to $b$ must contain at least two vertices in each of these lattice lines, two vertices in $b$, and one vertex strictly above $m$, giving a total of \[3 + 2 \lceil (k-1)/2 \rceil \geq 3 + 2 ( (k-1)/2 ) = k + 2\] vertices.
As the length of a path is the number of edges it contains, the path must have length at least $k+1$.
\end{proof}

\begin{lemma} \label{lem:bilen}
The $i$-th lattice line below and parallel to $e$ contains $h(i)$ gap locations between $\ell_1$ and $\ell_2$, where
\[i\sqrt{3} \tan\frac{\theta}{2} \leq h(i) \leq i\sqrt{3} \tan\frac{\theta}{2} + 2.\]
\end{lemma}
\begin{proof}
Let $b_i$ be the $i$-th lattice line below and parallel to~$e$. We use trigonometry to analyze the length of~$b_i$ between~$\ell_1$ and $\ell_2$; see Figure~\ref{fig:bilen}.
Consider the triangle formed by $b_i$,~$\ell_1$, and the line perpendicular to $e$ at $v_1$, which we call $\ell^*$.
Lines~$\ell_1$ and $\ell^*$ form an angle of $\theta/2$, and the distance between $e$ and~$b_i$ along $\ell^*$ is $i\sqrt{3} / 2$.
It follows that the length of $b_i$ between~$\ell_1$ and $\ell^*$ is $i\sqrt{3} \tan(\theta/2) / 2$.
Altogether, this implies~$b_i$ between $\ell_1$ and $\ell_2$ is of length $i\sqrt{3} \tan(\theta/2) + 1$.
As each edge of the triangular lattice has length 1, this means there are between $i\sqrt{3} \tan(\theta/2)$ and $i\sqrt{3} \tan(\theta/2) + 2$ gap locations in $b_i$, as claimed.
\end{proof}

\begin{figure}
\centering
\begin{subfigure}{.45\textwidth}
	\centering
	\includegraphics[scale = 0.74]{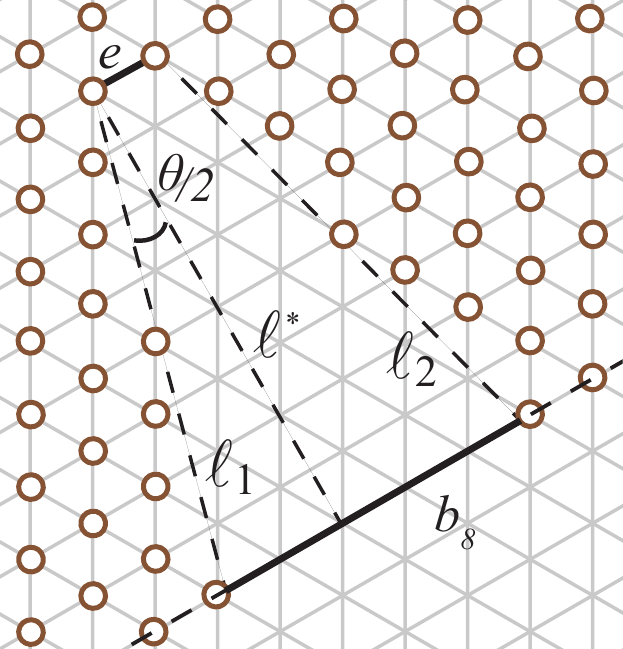}
	\caption{}
	\label{fig:bilen}
\end{subfigure}%
\begin{subfigure}{.45\textwidth}
	\centering
	\includegraphics[scale = 0.45]{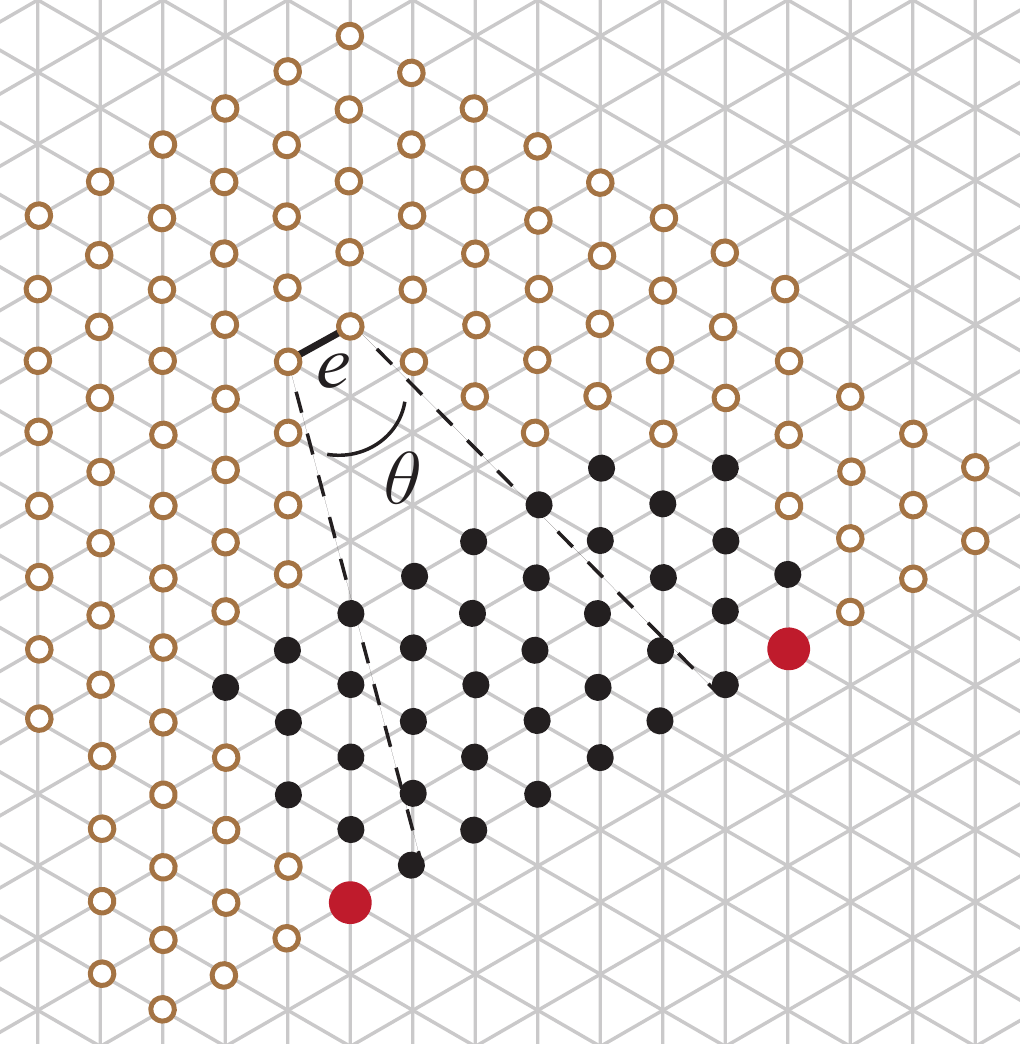}
	\caption{}
	\label{fig:sigma*}
\end{subfigure}
\caption{Figures from proofs in Section~\ref{subsec:smalltheta}. (a) A depiction of the notation used in the proof of Lemma~\ref{lem:bilen}; the intersection of $b_8$ and the gap is depicted as a solid segment, which is of length $8\sqrt{3} \tan(\theta/2)+1$ and contains 4 gap locations. (b) The configuration $\sigma^*$ used in Lemma~\ref{lem:Z1} for $\theta = \pi/6$ and $k = 8$.}
\label{fig:proofsupport}
\end{figure}

\begin{lemma} \label{lem:Z1}
Let $L$ be a V-shaped land mass of height $k$ and angle $\theta \leq \pi/3$. Then the normalizing constant $Z$ of the stationary distribution $\pi$ of $\M$	satisfies
\[Z \geq C \left[(\lambda\gamma)^{-2\sqrt{3}\tan\frac{\theta}{2}}\right]^k,\]
for a constant $C$ that depends on $\theta$, $\lambda$, and $\gamma$ but not on $k$.
\end{lemma}
\begin{proof}
Observe that $Z = \sum_{\sigma \in \Omega} \lambda^{-p(\sigma)}\gamma^{-g(\sigma)}$ satisfies $Z \geq \lambda^{-p(\sigma^*)}\gamma^{-g(\sigma^*)}$ for any $\sigma^* \in \Omega$.
We now construct a particular $\sigma^*$ (Figure~\ref{fig:sigma*}) and calculate its perimeter and gap perimeter.
Let $\sigma^*$ contain a straight line of particles along~$b$ connecting the two objects, and let $u$ be the number of objects and particles in this line.
By Lemma~\ref{lem:bilen}, since $b = b_k$ and $u$ includes two particles on land as well as two objects,
\[k\sqrt{3} \tan\frac{\theta}{2} + 4 \leq u \leq k\sqrt{3} \tan \frac{\theta}{2} + 6.\]

Continue constructing $\sigma^*$ by placing rows of $u$ particles above this initial row such that the row starts and ends on opposite sides of the gap.
By Lemma~\ref{lem:V_shape_initial_n}, there are ${4k+7}$ total objects and particles, so there will be $v = \lceil (4k+7)/u\rceil$ such rows, with the last row possibly incomplete.
We note that $v$ satisfies:
\begin{align*}
v = \left\lceil \frac{4k + 7}{u} \right\rceil &\leq \frac{4k + 7}{u} + 1 \leq \frac{4k + 7}{k\sqrt{3} \tan\frac{\theta}{2} + 4} + 1
\leq \frac{4}{\sqrt{3} \tan \frac{\theta}{2}} + \frac{7}{4} + 1
\leq \frac{4}{\sqrt{3} \tan \frac{\theta}{2} } + 3;\\
v = \left\lceil \frac{4k + 7}{u} \right\rceil &\geq \frac{4k+7}{u} \geq \frac{4k + 7}{k\sqrt{3} \tan\frac{\theta}{2} + 6}
\geq \frac{4k}{k\sqrt{3} \tan\frac{\theta}{2} + 6k}
\geq \frac{4}{\sqrt{3} \tan\frac{\theta}{2} + 6}.
\end{align*}

Configuration~$\sigma^*$ has perimeter at most $ 2u + 2v-4$ and gap perimeter at most $u - 4 + z$, where $z$ is the number of particles occupying gap locations in the upper perimeter of~$\sigma^*$.
These $z$ remaining particles must be in either the $(k - v + 1)$-th or $(k - v + 2)$-th lattice lines below $e$, so we can bound $z$ by again applying Lemma~\ref{lem:bilen}:
\[z \leq (k-v+1) \sqrt{3} \tan \frac{\theta}{2} + 2.\]
Altogether, this implies:
\[p(\sigma^*) \leq 2u + 2v - 4
\leq 2k\sqrt{3} \tan\frac{\theta}{2} + 12 + \frac{8}{\sqrt{3} \tan\frac{\theta}{2}} + 6 - 4
\leq k\left(2\sqrt{3} \tan\frac{\theta}{2}\right) + \left(\frac{8} {\sqrt{3} \tan\frac{\theta}{2}} + 14\right),\]
and
\begin{align*}
g(\sigma^*) &\leq u - 4 + z\\
&\leq k\sqrt{3} \tan\frac{\theta}{2} + 6 - 4 + (k - v + 1)\sqrt{3} \tan\frac{\theta}{2} + 2\\
&\leq 2k\sqrt{3} \tan\frac{\theta}{2} + \left(-\frac{4}{\sqrt{3}\tan\frac{\theta}{2} + 6} + 1\right) \sqrt{3} \tan\frac{\theta}{2} + 4\\
&\leq k\left(2\sqrt{3} \tan\frac{\theta}{2}\right) + \left(\sqrt{3}\tan\frac{\theta}{2} - \frac{4 \sqrt{3} \tan\frac{\theta}{2}}{\sqrt{3}\tan\frac{\theta}{2} + 6} + 4\right).
\end{align*}
We note that the second parentheses in the final bounds above for $p(\sigma^*)$ and $g(\sigma^*)$ are constants that only depend on $\theta$.
This implies that there is a constant
\[C = \lambda^{-\left(14 + \frac{8}{\sqrt{3} \tan\frac{\theta}{2}}\right)} \gamma^{-\left( \sqrt{3}\tan\frac{\theta}{2} -\frac{4 \sqrt{3} \tan\frac{\theta}{2}}{\sqrt{3}\tan\frac{\theta}{2} + 6} + 4\right)}\]
such that
\[Z \geq \lambda^{-p(\sigma^*)} \gamma^{-g(\sigma^*)} \geq C \left[\left(\lambda\gamma\right)^{-2\sqrt{3}\tan\frac{\theta}{2}}\right]^k.\]

As claimed, $C$ depends only on $\lambda$, $\gamma$, and $\theta$, and is independent of $k$.
\end{proof}

\begin{theorem} \label{thm:theta_1}
Let $\lambda > \cbound =: \cshort$ and $\gamma > 1$. Then there exists a constant $\theta_1$ such that for all V-shaped land masses with angle $\theta < \theta_1$, the probability that the inner perimeter is above midpoint $m$ is exponentially small in $k$, the height of the gap, provided $k$ is sufficiently large. In particular,
\[\theta_1 = 2\tan^{-1}\left(\frac{\log_{\lambda\gamma} \left(\lambda / \cshort\right)}{\sqrt{3}}\right).\]
\end{theorem}
\begin{proof}
Recall that $S_1 \subseteq \Omega$ is the set of configurations for which the inner perimeter is strictly above~$m$.
We show that~$S_1$ has exponentially small weight at stationarity; in particular, we show $\pi(S_1)$ is bounded above by $f_2(k) \xi^k$, where $f_2(k)$ is a subexponential function and $\xi < 1$ is a constant.

If $\sigma \in S_1$, then by Lemma~\ref{lem:inperimabovelen} we have $p(\sigma) \geq 2k+2$, as its inner perimeter --- and thus the rest of the perimeter as well --- must be above $m$.
Furthermore, because the perimeter by definition includes both objects and particles, which number $4k+7$ by Lemma~\ref{lem:V_shape_initial_n}, any configuration $\sigma \in \Omega$ has $p(\sigma) \leq 2(4k + 7) - 2 = 8k + 12$.
A result from~\cite{Cannon2016} exploits a connection to self-avoiding walks in the hexagon lattice to show the number of connected, hole-free particle configurations with perimeter $p$ is at most $f(p)(\cshort)^p$ for some subexponential function $f$. This is certainly also an upper bound on the number of configurations in $S_1$ with perimeter $p$.
Because ${\gamma^{-g(\sigma)} < 1}$, we have:
\[\pi(S_1) = \sum_{\sigma \in S_1} \frac{\lambda^{-p(\sigma)} \gamma^{-g(\sigma)}}{Z} < \sum_{p = 2k + 2}^{8k + 12} \frac{f(p)\cshort^p\lambda^{-p}}{Z}.\]
Let $f_1(k) = \sum_{p = 2k + 2}^{8k + 12}f(p)$, and note that this function is subexponential in~$k$ because its number of summands is linear in~$k$.
Because $\lambda > \cshort$ and $p \geq 2k + 2$, we have that:
\[\pi(S_1) \leq \frac{f_1(k) \left(\frac{\cshort}{\lambda}\right)^{2k + 2}}{Z}.\]

By Lemma~\ref{lem:Z1}, there is a constant $C_1 = \cshort^2/(\lambda^2 C)$ such that:
\[\pi(S_1) \leq \frac{f_1(k) \left(\frac{\cshort}{\lambda}\right)^{2k + 2}}{C \left[ \left(\lambda\gamma\right)^{-2\sqrt{3}\tan\frac{\theta}{2}} \right]^k} = C_1f_1(k) \left(\frac{\cshort(\lambda\gamma)^{\sqrt{3}\tan\frac{\theta}{2}}}{\lambda}\right)^{2k}.\]
For all $\theta < 2 \tan^{-1}\left(\log_{\lambda\gamma} (\lambda/\cshort) / \sqrt{3}\right)$, the term in parentheses above is less than one:
\[\frac{\cshort(\lambda\gamma)^{\sqrt{3}\tan\frac{\theta}{2}}}{\lambda} < \frac{\cshort(\lambda\gamma)^{\log_{\lambda\gamma}\left(\frac{\lambda}{\cbound}\right)}}{\lambda} = 1.\]
Because $C_1f_1(k)$ is a subexponential function but the term above, raised to the $2k$ power, is exponentially small, the latter eventually dominates and we conclude there is a constant $\xi < 1$ such that for sufficiently large $k$, $\pi(S_1) < \xi^k$, proving the theorem.
\end{proof}

Since $n = 4k+5$ by Lemma~\ref{lem:V_shape_initial_n}, the probability that the inner perimeter is above point $m$ is also exponentially small in $n$, the number of particles.

As an example, for $\lambda = 4$ and $\gamma = 2$ (the parameters of the simulations in Figure~\ref{fig:theta_dependence} and Figure~\ref{fig:antbridge}), our methods give $\theta_1 = 0.0879 \sim 5.03^\circ$.
However, simulations suggest this bound is far from tight.
In general, as $\lambda$ increases, so does the angle~$\theta_1$: a stronger bias towards a shorter perimeter means the bridge forms closer to the bottom of the gap and at even larger angles the bridge remains below $m$.
Similarly, as $\gamma$ decreases the bridge moves down towards the bottom of the gap and at even larger angles remains below $m$.

As with Theorem~\ref{thm:alphabridge}, we are unable to give explicit bounds on the ``sufficiently large $k$" required by the statement of Theorem~\ref{thm:theta_1} because determining the exact form of the subexponential function $f(p)$ in the above proof remains an open problem (see Section 4 of~\cite{Duminil-Copin2012}).
However, we expect and observe that the claims of this theorem hold even for the small $k$ for which our proofs do not apply.

\begin{figure}
\centering
\includegraphics[scale = 0.75]{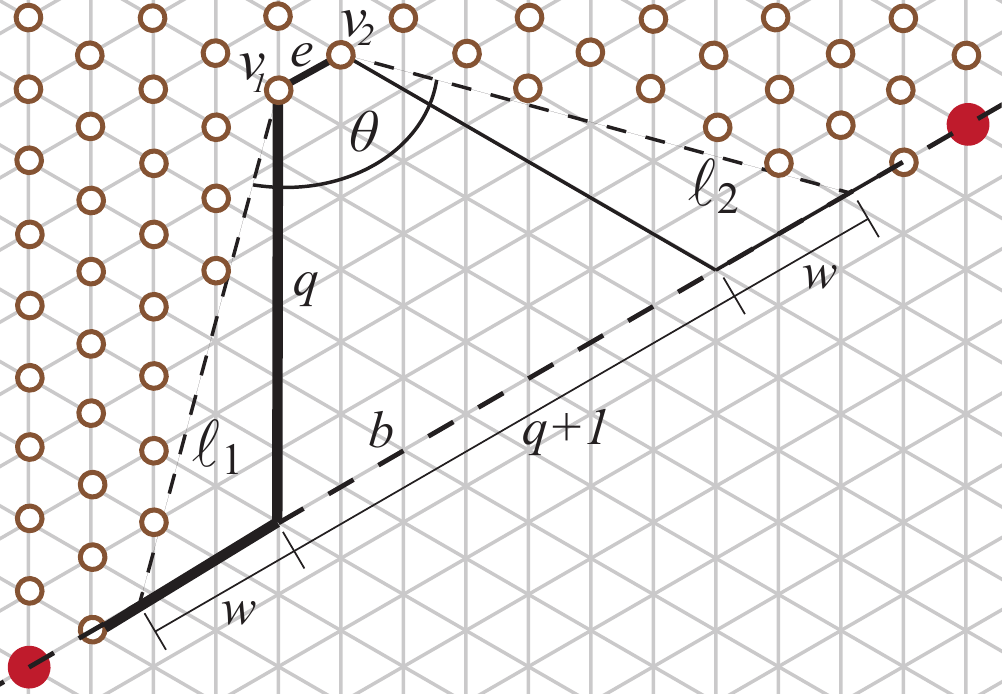}
\caption{The path of length $k$ (bold) from vertex $v_1$ to the first land location in line $b$ considered in the proof of Lemma~\ref{lem:q}; this path is used to calculate the gap height $k$ in terms of the gap depth $q$. By also considering the reflection of this path from $v_2$ (solid line), we can calculate the distance between the two objects to be $q + 2\lceil w \rceil + 3$ (Lemma~\ref{lem:p_min}).}
\label{fig:q-relationships}
\end{figure}


\subsection{Proofs for Large $\theta$} \label{subsec:largetheta}

We now consider the set $S_2 = \Omega \setminus S_1$, which consists of all configurations where the inner perimeter goes through or below $m$.
We will show that for some large angles $\theta$, for all $\lambda > \cbound$ and $\gamma > (\cbound)^4 \lambda^4$, $\pi(S_2)$ is exponentially small.
While a lower bound on $\gamma$ is necessary for the proofs presented below, we believe this is an artifact of our proof rather than the problem itself and suspect this requirement can be loosened or removed altogether.

For $\theta \geq \pi/3$, it is no longer true that a V-shaped land mass of height $k$ has exactly $k-1$ lattice lines between $b$ and~$e$.
We define a new quantity $q$, the \emph{gap depth}, as the length of a shortest path from $e$ to $b$ in $\Gtri$; unlike in the definition of the height $k$ of a gap, this shortest path is not required to stay on land locations.
The Euclidean distance between $e$ and $b$ is then $\sqrt{3}q / 2$.
Furthermore, $q$ can be expressed as a function of $k$ and $\theta$.

\begin{lemma}\label{lem:q}
For a V-shaped land mass of height $k$ and angle $\theta \geq \pi/3$, the gap depth $q$ satisfies
\[k = \left\lceil \left(\frac{1}{2} + \frac{\sqrt{3}}{2} \tan\frac{\theta}{2} \right) q \right\rceil.\]
\end{lemma}
\begin{proof}
Consider the path from $v_1$ to line $b$ that leaves $v_1$ forming an angle of $2\pi/3$ with $e$, and then proceeds along $b$ until it reaches a land location; see Figure~\ref{fig:q-relationships}, where this path is shown in bold.
The total length of this path is $k$, and its first segment from $v_1$ to $b$ is length~$q$.
Let $w$ be the length of $b$ between this path's turning point and~$\ell_1$; then $k = q + \lceil w \rceil$.
This path and~$\ell_1$ form an obtuse triangle where two sides have lengths~$q$ and $w$, respectively.
The angle opposite the side of length $w$ is $\theta/2 - \pi/6$, while the angle opposite the side of length $q$ is $\pi - 2\pi/3 - (\theta/2-\pi/6) = \pi/2 - \theta/2$.
Length $w$ can be calculated in terms of length $q$ with the law of sines:
\[w = \frac{\sin\left(\frac{\theta}{2} - \frac{\pi}{6}\right)}{\sin \left(\frac{\pi}{2} - \frac{\theta}{2}\right)}\ q
= \frac{\sin\frac{\theta}{2} \cos\frac{\pi}{6} - \cos\frac{\theta}{2} \sin\frac{\pi}{6}}{\cos\frac{\theta}{2}}\ q
= \frac{\frac{\sqrt{3}}{2} \sin\frac{\theta}{2} - \frac{1}{2}\cos\frac{\theta}{2}}{\cos\frac{\theta}{2}}\ q
= \frac{q\sqrt{3}}{2} \tan\frac{\theta}{2} - \frac{q}{2}.\]

Because $q$ is an integer, it follows that
\[k = q + \lceil w \rceil
= \left\lceil q + \frac{q\sqrt{3}}{2} \tan\frac{\theta}{2} -\frac{q}{2} \right\rceil
= \left\lceil \left(\frac{1}{2} + \frac{\sqrt{3}}{2} \tan\frac{\theta}{2} \right)q \right\rceil,\]
which is the desired result.
\end{proof}

For simplicity, we do the bulk of our analysis using $q$ instead of $k$.
The previous lemma shows that proving an expression is exponentially small in $q$ implies it is also exponentially small in $k$.

\begin{lemma} \label{lem:p_min}
For any V-shaped land mass of gap depth $q$ and angle $\theta \geq \pi/3$, any configuration $\sigma$ has perimeter at least
\[p(\sigma) \geq \left(2\sqrt{3} \tan\frac{\theta}{2}\right)q + 6.\]
\end{lemma}
\begin{proof}
We first bound the distance between the two objects on either side of the gap.
Using the length $w$ from the proof of Lemma~\ref{lem:q}, the distance between the two objects in any configuration is $q + 2\lceil w \rceil + 3 \geq q + 2w + 3$ (see Figure~\ref{fig:q-relationships}).
The perimeter of any particle configuration is at least twice this distance, so for any $\sigma$,
\[p(\sigma) \geq 2q + 4w + 6
= 2q + 4 \left(\frac{q\sqrt{3}}{2} \tan\frac{\theta}{2} -\frac{q}{2} \right) + 6
= \left(2\sqrt{3}\tan\frac{\theta}{2} \right)q + 6,\]
which is the desired bound.
\end{proof}

\begin{figure}
\centering
\begin{subfigure}{.45\textwidth}
	\centering
    \includegraphics[scale = 0.65]{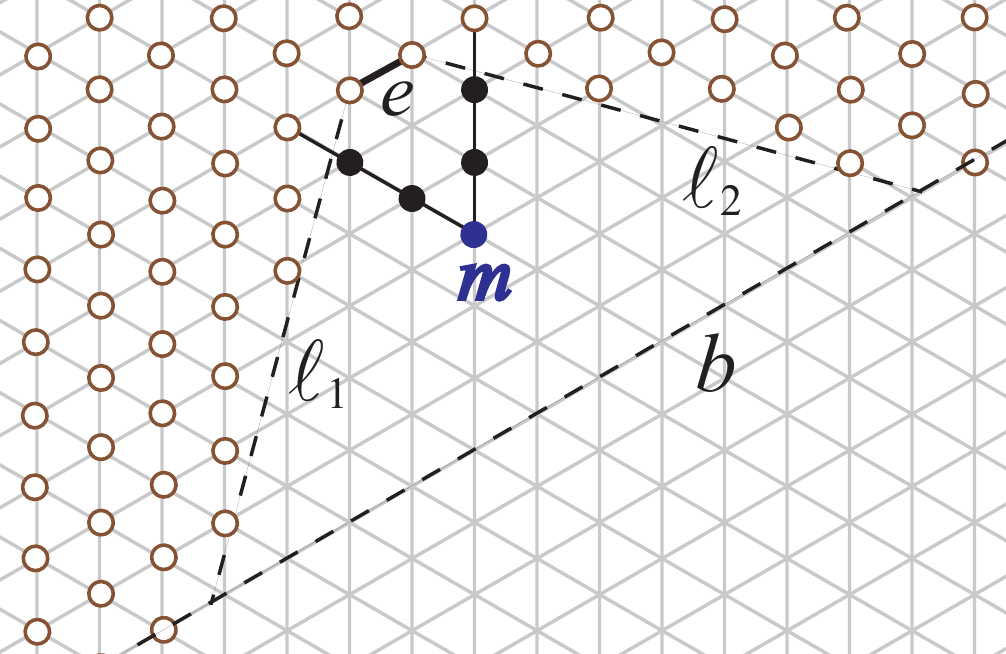}
    \caption{}
    \label{fig:shortpath}
\end{subfigure}%
\begin{subfigure}{.45\textwidth}
	\centering
    \includegraphics[scale = 0.88]{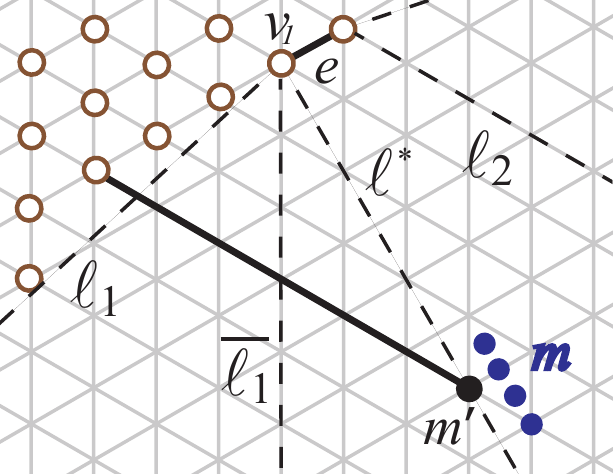}
    \caption{}
    \label{fig:pathintersect}
\end{subfigure}
\caption{From the proof of Lemma~\ref{lem:in_gap}: (a) An example of a shortest path between land locations on opposite sides of the gap passing through midpoint $m$. (b) The four possible locations for midpoint $m$ for which a shortest path passing through or below $m$ contains $m'$, and a shortest path from $m'$ to a land location (solid line).}
\label{fig:max_inner_perim_m}
\end{figure}

\begin{lemma} \label{lem:in_gap}
For any V-shaped land mass of gap depth $q$ and angle $\theta > \pi/3$, any configuration $\sigma \in S_2$ (passing below or through midpoint $m$ of the gap) has gap perimeter $g(\sigma) \geq \frac{q}{2}$.
\end{lemma}
\begin{proof}
If $\sigma \in S_2$, i.e., if its inner perimeter passes through or below $m$, then it must contain a path that starts and ends at land locations and also passes through or below $m$.
We consider all such paths and give a lower bound on the number of gap locations they must contain.
The shortest such paths start and end on opposite sides of the gap, so we focus on paths of this type.

If $m$ is a vertex of $\Gtri$, one shortest path between land locations passing through $m$ leaves $m$ along the two lattice lines not parallel to $e$ and follows them until reaching the land mass, as in Figure~\ref{fig:shortpath}.
If $m$ is on a lattice edge, a shortest path passing below $m$ is constructed in the same way, beginning from each of the edge's endpoints.
Otherwise, if $m$ is neither a lattice point nor on a lattice edge, the same procedure is followed for the first lattice point or lattice edge below $m$.
In all cases, let $m'$ be the point of intersection between this path and $\ell^*$, the line perpendicular to~$e$ through~$v_1$.
Figure~\ref{fig:pathintersect} shows all the possible locations of~$m$ producing a particular~$m'$.
Inspection shows that in all of these cases,~$m'$ is contained in the $2\lfloor \frac{q+1}{4} \rfloor$-th lattice line below~$e$.

Let $\overline{\ell_1}$ be the line from $v_1$ to $b$ forming an angle of $2\pi/3$ with $e$; see Figure~\ref{fig:pathintersect}.
Because $\theta > \pi/3$, all vertices of $\Gtri$ contained in $\overline{\ell_1}$, except $v_1$, are gap locations.
Any shortest path from $m'$ to a land location must share a vertex of $\Gtri$ with line $\overline{\ell_1}$.
Because $m'$ is in the $2\lfloor \frac{q+1}{4} \rfloor$-th lattice line below~$e$, any path from $m'$ to $\overline{\ell_1}$ is of length at least $\lfloor \frac{q+1}{4} \rfloor$ and contains at least $\lfloor \frac{q+1}{4} \rfloor + 1$ gap locations, including both of its endpoints.
By symmetry, this means any path between land locations passing below $m$, and thus any inner perimeter of a particle configuration passing below $m$, contains at least
\[2\left(\left\lfloor \frac{q+1}{4} \right\rfloor + 1\right) \geq 2\left(\frac{q-2}{4} + 1\right) \geq \frac{q}{2}\]
gap locations, as claimed.
\end{proof}

\begin{theorem}	\label{thm:theta_2}
Let $\lambda > \cbound =: \cshort$ and $\gamma > (\lambda\cshort)^4$. Then there exists a constant $\theta_2 > \pi/3$ such that for all V-shaped land masses with angle $\theta \in (\pi/3, \theta_2)$, the probability that the inner perimeter goes through or below midpoint $m$ is exponentially small in $k$, the height of the gap, provided $k$ is sufficiently large.
\end{theorem}
\begin{proof}
Recall $S_2$ is the set of all configurations whose inner perimeter goes through or below $m$.
We show that $\pi(S_2)$ is exponentially small in $k$, the height of the gap. By definition,
\[\pi(S_2) = \frac{\sum_{\sigma \in S_2} \lambda^{-p(\sigma)}\gamma^{-g(\sigma)}}{Z}.\]

By Lemma~\ref{lem:V_shape_initial_n}, the number of particles and objects in $\sigma_0$ for a land mass of height $k$ is $4k+7$.
Since $\sigma_0$ is a path of width 2 and every particle occupies a land location, $p(\sigma_0) = 4k+7$ and $g(\sigma_0) = 0$. Thus,
\[Z = \sum_{\sigma \in \Omega} \lambda^{-p(\sigma)}\gamma^{-g(\sigma)} \geq \lambda^{-p(\sigma_0)}\gamma^{-g(\sigma_0)} = \lambda^{-4k-7}.\]
It is simpler to work with gap depth $q$ instead of gap height~$k$.
By Lemma~\ref{lem:q}, $k$ satisfies $k \leq \left(\frac{1}{2} + \frac{\sqrt{3}}{2}\tan\frac{\theta}{2}\right)q + 1$, so
\[Z \geq \lambda^{-4k-7} \geq \lambda^{-4\left(\frac{1}{2} + \frac{\sqrt{3}}{2} \tan\frac{\theta}{2}\right)q - 4 - 7} = \lambda^{-\left(2 + 2\sqrt{3}\tan\frac{\theta}{2}\right)q - 11}.\]
Combining this with Lemma~\ref{lem:in_gap},
\[\pi(S_2) = \sum_{\sigma \in S_2} \frac{\lambda^{-p(\sigma)}\gamma^{-g(\sigma)}}{Z}
\leq \lambda^{\left(2 + 2\sqrt{3} \tan\frac{\theta}{2} \right)q + 11}\sum_{\sigma \in S_2}\lambda^{-p(\sigma)}\gamma^{-\frac{q}{2}}.\]

Let $p_{min}$ (resp., $p_{max}$) be the minimum (resp., maximum) possible perimeter for a valid particle configuration in $S_2$.
By Lemma~\ref{lem:p_min}, $p_{min} \geq 2\sqrt{3}\tan(\theta/2) q$.
As shown in the proof of Theorem~\ref{thm:theta_1}, $p_{max} = 8k+12$; in terms of $q$, by Lemma~\ref{lem:q},
\[p_{max} \leq 8\left(\frac{q}{2} + \frac{q\sqrt{3}}{2}\tan\frac{\theta}{2} + 1\right) + 12
= 4q + 4q\sqrt{3}\tan\frac{\theta}{2} + 20.\]

Using the result from~\cite{Cannon2016} which upper bounds the number of particle configurations with perimeter $p$ by the expression $f(p)\cshort^p$, for some subexponential function $f$, we have that:
\begin{align*}
\pi(S_2) &\leq \lambda^{\left(2 + 2\sqrt{3}\tan\frac{\theta}{2}\right)q + 11} \sum_{p=p_{min}}^{p_{max}} f(p)\cshort^p\lambda^{-p}\gamma^{-\frac{q}{2}}\\
&\leq \lambda^{\left(2 + 2\sqrt{3} \tan\frac{\theta}{2}\right)q + 11} \left(\sum_{p=p_{min}}^{p_{max}} f(p)\right) \left(\frac{\cshort}{\lambda}\right)^{p_{min}} \gamma^{-\frac{q}{2}}\\
&\leq \left(\lambda^{11} \sum_{p = p_{min}}^{p_{max}} f(p)\right) \left(\lambda^{\left(2 + 2\sqrt{3} \tan\frac{\theta}{2}\right)} \left(\frac{\cshort}{\lambda}\right)^{2\sqrt{3}\tan\frac{\theta}{2}} \gamma^{-\frac{1}{2}} \right)^q\\
&= \left(\lambda^{11} \sum_{p = p_{min}}^{p_{max}} f(p)\right) \left(\lambda^2 \cshort^{2\sqrt{3}\tan\frac{\theta}{2}}\gamma^{-\frac{1}{2}}\right)^q.
\end{align*}

The first parentheses is a function $f_1(q)$ that is subexponential in $q$, as it has a polynomial number of summands based on our calculations of $p_{min}$ and $p_{max}$ (which are expressions in terms of $q$), and each summand is subexponential.
When the term in the second set of parentheses above is less than one, the second factor (this term raised to the $q$ power) is exponentially small in $q$, the gap depth, and thus for sufficiently large $q$ this term dominates and the entire expression is exponentially small in $q$.
This holds whenever $\theta$ satisfies:
\[\theta < 2\tan^{-1}\left(\frac{1}{2\sqrt{3}} \log_{\cshort}\left(\gamma^{1/2} \lambda^{-2}\right) \right)
= 2\tan^{-1}\left(\frac{1}{\sqrt{3}} \log_{\cshort}\left(\frac{\gamma^{1/4}}{\lambda}\right) \right) =: \theta_2.\]

Whenever $\gamma^{1/4}/\lambda > \cshort$ --- i.e., whenever $\gamma > (\lambda\cshort)^4$ --- the argument of $\tan^{-1}$ above is at least $1/\sqrt{3}$, and thus $\theta_2 > \pi/3$.
It follows that whenever $\gamma > (\lambda\cshort)^4$ and $\theta \in (\pi/3,\theta_2)$,
\[\pi(S_2) < f_1(q)\psi^q,\]
where $f_1(q)$ is subexponentially large in $q$ and $\psi < 1$ so the second term is exponentially small in $q$.
For sufficiently large $q$, the second term dominates, and we conclude the weight of set $S_2$ at stationarity is exponentially small in $q$.
Because $k$ and $q$ differ only by additive and multiplicative constants, it is also exponentially small in $k$, the gap height, for sufficiently large $k$.
\end{proof}

As was the case for small angles, here also we have that by Lemma~\ref{lem:V_shape_initial_n}, there are $n = 4k + 5$ particles.
Thus, we have that the probability the inner perimeter goes through or below midpoint $m$ when $\theta$ is sufficiently large is also exponentially small in $n$.

If we again use the example value of $\lambda = 4$ (as in the simulations depicted in Figure~\ref{fig:antbridge} and Figure~\ref{fig:theta_dependence}), Theorem~\ref{thm:theta_2} requires $\gamma > (\lambda\cshort)^4 \approx 34,786$.
This value is large, but importantly is constant (i.e., it does not depend on $n$) and is only an artifact of our proof.
For example, when $\lambda = 4$ and $\gamma = 10^5$, our methods show that the resulting bridge remains above midpoint $m$ with high probability for any angle between $\pi / 3 = 60^\circ$ and $\theta_2 \approx 1.2234 \sim 70.10^\circ$.
On the other hand, an experiment with $\lambda = 4$, $\gamma = 2$, and $\theta = 90^\circ$ is shown in Figure~\ref{fig:theta_dependence_90} to remain well above the midpoint $m$, suggesting that this behavior is stable for much smaller values of $\gamma$ and a much larger range of angles than we were able to prove.

As for Theorems~\ref{thm:alphabridge} and~\ref{thm:theta_1}, we are unable to give explicit bounds on the ``sufficiently large $k$" required by the statement of Theorem~\ref{thm:theta_2} because the exact form of $f(p)$ in its proof is unknown, but we expect and observe that it holds even for the small $k$ for which our proof does not apply.

\section{Simulations} \label{sec:simulations}

\begin{figure*}[t]
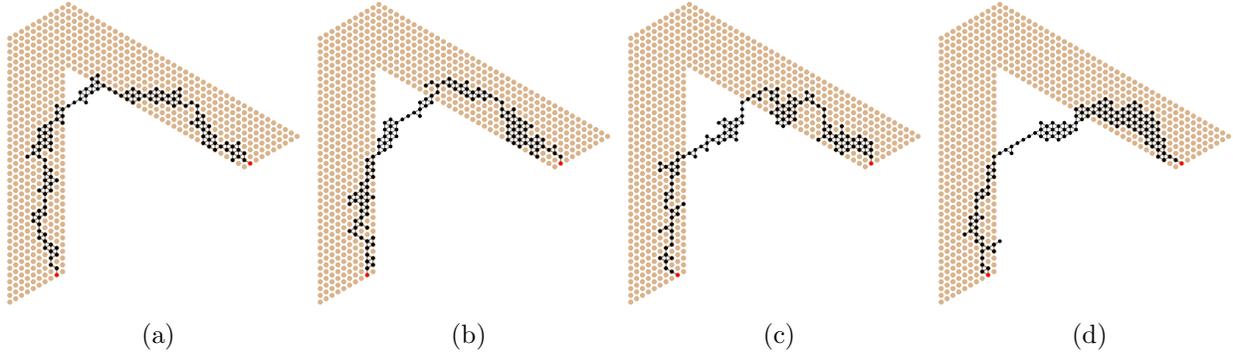

\centering
\begin{subfigure}{.25\textwidth}
	\centering
	\iftikz\input{fig_angle60_lambda4_gamma2_2000000.tex}\fi
	\caption{}
\end{subfigure}%
\begin{subfigure}{.25\textwidth}
	\centering
	\iftikz\input{fig_angle60_lambda4_gamma2_4000000.tex}\fi
	\caption{}
\end{subfigure}%
\begin{subfigure}{.25\textwidth}
	\centering
	\iftikz\input{fig_angle60_lambda4_gamma2_6000000.tex}\fi
	\caption{}
\end{subfigure}%
\begin{subfigure}{.25\textwidth}
	\centering
	\iftikz\input{fig_angle60_lambda4_gamma2_8000000.tex}\fi
	\caption{}
\end{subfigure}
\caption{A particle system using biases $\lambda = 4$ and $\gamma = 2$ to shortcut a V-shaped land mass with $\theta = \pi/3$ after (a) 2 million, (b) 4 million, (c) 6 million, and (d) 8 million iterations of Markov chain $\M$, beginning in configuration $\sigma_0$ shown in Figure~\ref{fig:initconfig_V}.}
\label{fig:antbridge}
\end{figure*}

\begin{figure*}[t]
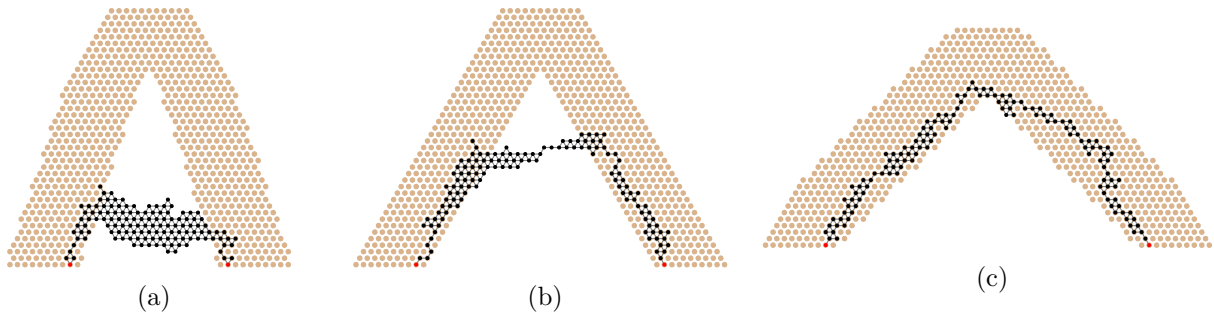

\centering
\begin{subfigure}{.28\textwidth}
	\centering
	\iftikz\input{fig_30_20000000.tex}\fi
	\caption{}
\end{subfigure}%
\begin{subfigure}{.35\textwidth}
	\centering
	\iftikz\input{fig_60_20000000_smaller.tex}\fi
	\caption{}
\end{subfigure}%
\begin{subfigure}{.37\textwidth}
	\centering
	\iftikz\input{fig_90_20000000.tex}\fi
	\caption{}
	\label{fig:theta_dependence_90}
\end{subfigure}
\caption{A particle system using biases $\lambda = 4$ and $\gamma = 2$ to shortcut a V-shaped land mass with angle (a) $\pi/6$, (b) $\pi/3$, and (c) $\pi/2$ after 20 million iterations of Markov chain $\M$. For a given angle, the land mass $L$ and initial configuration $\sigma_0$ were constructed as described in Section~\ref{sec:proofs}.}
\label{fig:theta_dependence}
\end{figure*}

\begin{figure*}[t]
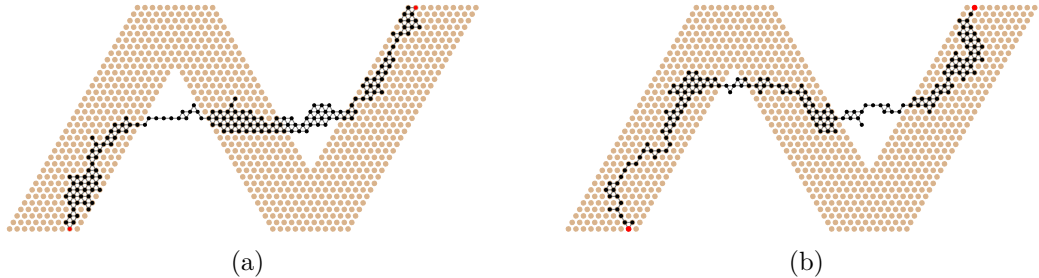

\centering
\begin{subfigure}{.45\textwidth}
	\centering
	\iftikz\input{fig_N_10000000.tex}\fi
	\caption{}
\end{subfigure}%
\begin{subfigure}{.45\textwidth}
	\centering
	\iftikz\input{fig_N_20000000.tex}\fi
	\caption{}
\end{subfigure}
\caption{A particle system using biases $\lambda = 4$ and $\gamma = 2$ to shortcut an N-shaped land mass after (a) 10 million and (b) 20 million iterations of Markov chain $\M$, beginning in configuration $\sigma_0$ shown Figure~\ref{fig:initconfig_N}.}
\label{fig:Nsim}
\end{figure*}

We can see the performance of our algorithm from simulation results on a variety of instances.
Figure~\ref{fig:antbridge} shows snapshots over time for a bridge shortcutting a V-shaped gap with internal angle $\theta = \pi/3$ and biases $\lambda = 4, \gamma = 2$.
Qualitatively, this bridge matches the shape and position of the army ant bridges in~\cite{Reid2015}.
Figure~\ref{fig:theta_dependence} shows the results of an experiment that held $\lambda$, $\gamma$, and the number of iterations of $\M$ constant, varying only the internal angle of the V-shaped land mass.
The particle system exhibits behavior consistent with the theoretical results in Section~\ref{sec:proofs} and the army ant bridges in~\cite{Reid2015}, shortcutting closer to the bottom of the gap when $\theta$ is small and staying almost entirely on land when $\theta$ is large.
Lastly, Figure~\ref{fig:Nsim} shows the resulting bridge structure when the land mass is N-shaped, demonstrating that our algorithm can be generalized beyond the original inspiration of V-shaped land masses to shortcut multiple gaps in more complex structures.

These simulations demonstrate the successful application of our stochastic approach to shortcut bridging.
Moreover, experimenting with variants suggests this approach may be useful for other related applications in the future.

\section{Conclusions and Future Directions} \label{sec:futurework}

In summary, we presented a Markov chain $\M$ that can be directly translated to a stochastic, distributed, local, asynchronous algorithm $\A$ that provably solves the shortcut bridging problem.
Furthermore, in the special case of bridging over the gap in a V-shaped land mass, we rigorously analyzed the effect of the land mass's internal angle, showing that below one threshold angle the bridge will shortcut near the bottom of the gap, and above another threshold angle the bridge will remain close to land, with high probability.

Several directions of further investigation seem promising.
The successful application of our stochastic approach to shortcut bridging suggests it may be useful for other types of problems as well; one related behavior of particular interest is ``exploration bridging'', where a particle system first explores its environment to discover sites of interest, and then converges to a bridge-like structure between them.
We are also interested in formulating alternative local rules for shortcut bridging which yield bridges that appear more ``structurally sound,'' though we suspect the information needed to do so may be difficult to encode in our particle systems~due to the constant-size memory constraint of the amoebot model.

%

\bibliographystyle{abbrv}
\bibliography{ref}

\end{document}